\documentclass[12pt, draftclsnofoot, onecolumn]{IEEEtran}
\ifCLASSINFOpdf

\else

\fi

\usepackage{amsmath}
\usepackage{makeidx}  
\usepackage{algorithm}
\usepackage{algorithmic}
\usepackage{graphicx}
\usepackage{subfigure}
\usepackage{epstopdf}
\usepackage{bm}
\usepackage{bbding}
\usepackage{cite}
\usepackage{stfloats}

\usepackage{amssymb}
\usepackage{graphicx}

\usepackage{url}

\usepackage{tabularx,booktabs}
\newcolumntype{C}{>{\centering\arraybackslash}X} 
\usepackage{lipsum}

\usepackage{makecell} 

\usepackage{graphicx}
\usepackage{color}
\newtheorem{thm}{Theorem}
\newtheorem{rem}{Remark}

\newtheorem{pos}{Proposition}
\newtheorem{proof}{proof}

\begin{document}
\linespread{1.35}

\title{Intelligent Reflecting Surface Aided MIMO Networks: Distributed or Centralized Architecture?}
\author{Guangji~Chen,\IEEEmembership{}
        Qingqing~Wu,\IEEEmembership{}
        Wen~Chen,\IEEEmembership{}
        Yanzhao~Hou,\IEEEmembership{}
        Mengnan~Jian,\IEEEmembership{}
        Shunqing~Zhang,\IEEEmembership{}
        and Jun~Li \vspace{-20pt}
        \thanks{Guangji Chen and Jun Li are with Nanjing University of Science and Technology, Nanjing 210094, China (email: guangjichen@njust.edu.cn; jun.li@njust.edu.cn). Qingqing Wu and Wen Chen are with Shanghai Jiao Tong University, 200240, China (e-mail: qingqingwu@sjtu.edu.cn; wenchen@sjtu.edu.cn). Yanzhao Hou is with Beijing University of Posts and Telecommunications, Beijing 100876, China (houyanzhao@bupt.edu.cn). Mengnan Jian is with ZTE Corporation, Shenzhen 518057, China (e-mail: jian.mengnan@zte.com.cn). Shunqing Zhang is with Shanghai University, Shanghai 20044 China (e-mail: shunqing@shu.edu.cn).}}

\maketitle
\vspace{-28pt}
\begin{abstract}
Intelligent reflecting surfaces (IRSs) have recently attained growing popularity in wireless networks owning to their capability to customize the wireless channel via smartly configured passive reflections. In addition to optimizing IRS reflection patterns, the flexible deployment of IRSs offers another design degree of freedom (DoF) to reconfigure the wireless propagation environment in favour of signal transmission. To unveil the impact of IRS deployment on the system capacity, we investigate the capacity of a broadcast channel with a multi-antenna base station (BS) sending independent messages to multiple users, aided by IRSs with $N$ elements. In particular, both the \emph{distributed} and \emph{centralized IRS} deployment architectures are considered. Regarding the \emph{distributed IRS}, the $N$ IRS elements form multiple IRSs and each of them is installed near a user cluster; while for the centralized IRS, all IRS elements are located in the vicinity of the BS. To draw essential insights, we first derive the maximum capacity achieved by the \emph{distributed IRS} and \emph{centralized IRS}, respectively, under the assumption of line-of-sight propagation and homogeneous channel setups. By capturing the fundamental tradeoff between the spatial multiplexing gain and passive beamforming gain, we rigourously prove that the capacity of the \emph{distributed IRS} is higher than that of the \emph{centralized IRS} provided that the total number of IRS elements is above a threshold. Motivated by the superiority of the \emph{distributed IRS}, we then focus on the transmission and element allocation design under the \emph{distributed IRS}. By exploiting the user channel correlation of intra-clusters and inter-clusters, an efficient hybrid multiple access scheme relying on both spatial and time domains is proposed to fully exploit both the passive beamforming gain and spatial DoF. Moreover, the IRS element allocation problem is investigated for the objectives of sum-rate maximization and minimum user rate maximization, respectively. Finally, extensive numerical results are provided to validate our theoretical finding and also to unveil the effectiveness of the
distributed IRS for improving the system capacity under various system setups.

\end{abstract}


\vspace{-8pt}
\section{Introduction}
The forthcoming sixth-generation (6G) wireless network is expected to satisfy the ever-growing demand for higher system capacity, enhanced reliability, and reduced latency \cite{saad2019vision}. In view of this issue, intelligent reflecting surfaces (IRSs) have been proposed as an appealing candidate for 6G owning to their potential to customize the wireless channel via passive reflection \cite{8910627,di2020smart,wu2021intelligent}. Briefly, an IRS is a controllable surface comprising a large number of tunable passive reflecting elements \cite{8910627}. By independently adjusting the phase-shift and amplitude of each reflecting element, a ``smart radio environment'' can be created, thereby strengthing the desired signals and mitigating the interference \cite{di2020smart}. Furthermore, IRSs enjoy additional practical advantages such as conformal geometry, light weight, and low profile, hence they can be conveniently deployed in future wireless networks for coverage enhancement. Owning to their appealing features, extensive researches have been carried out to facilitate the integration of IRSs into 6G wireless networks, by addressing their practical challenges, including channel estimation, IRS phase-shift optimization, and IRS deployment architecture/placement design. (see \cite{wu2021intelligent} and the references therein).

To fully reap the potential gains brought about by the IRS, it is of paramount importance to appropriately design its reflection coefficients so that the wireless propagation environment is reshaped for favorable signal transmission. This benefit has spurred great enthusiasm in the community. Motivated by this, the refection pattern design of the IRS has been extensively studied in the literature under various system setups, such as non-orthogonal multiple access (NOMA) \cite{mu2020exploiting, chen2021irs1, zheng2020intelligent, chen2022active, fu2021reconfigurable}, orthogonal frequency division multiplexing (OFDM) based wireless systems \cite{yang2020intelligent, li2021intelligent}, multiple-input multiple-output (MIMO) systems \cite{10186460, 9459505, chen2020performance, 9279253, pan2020multicell}, integrated sensing and communication \cite{hua2022joint, meng2022intelligent}, as well as wireless information and power transfer \cite{9610992,zhi2022active,9716123,9400380}. In particular, from an information theoretical viewpoint, it is essential to characterize the fundamental limits of IRS aided wireless systems so as to understand their achievable maximum performance gains. To this end, the seminal work \cite{wu2019beamforming} investigated the link-level received power and unveiled that a passive beamforming gain of order ${\cal O}\left( {{N^2}} \right)$ can be achieved for a total number of $N$ elements. Regarding a single-user MIMO system, the authors of \cite{zhang2020capacity} studied the maximum capacity achieved by jointly optimizing the IRS-aided passive beamforming and active transceiver/receiver beamforming. For IRS-aided multi-user systems, the capacity regions of the broadcast channel were characterized in \cite{mu2021capacity, 10159991} under the single-antenna and multi-antenna setups, respectively. Owing to the favorable channels brought by the IRS, the results in \cite{10159991} demonstrated that the capacity achieved by the efficient linear transmit precoder is able to approach that of classic dirty paper coding (DPC) under large $N$.
\begin{figure}[!t]
\setlength{\abovecaptionskip}{-5pt}
\setlength{\belowcaptionskip}{-5pt}
\centering
\includegraphics[width= 0.75\textwidth]{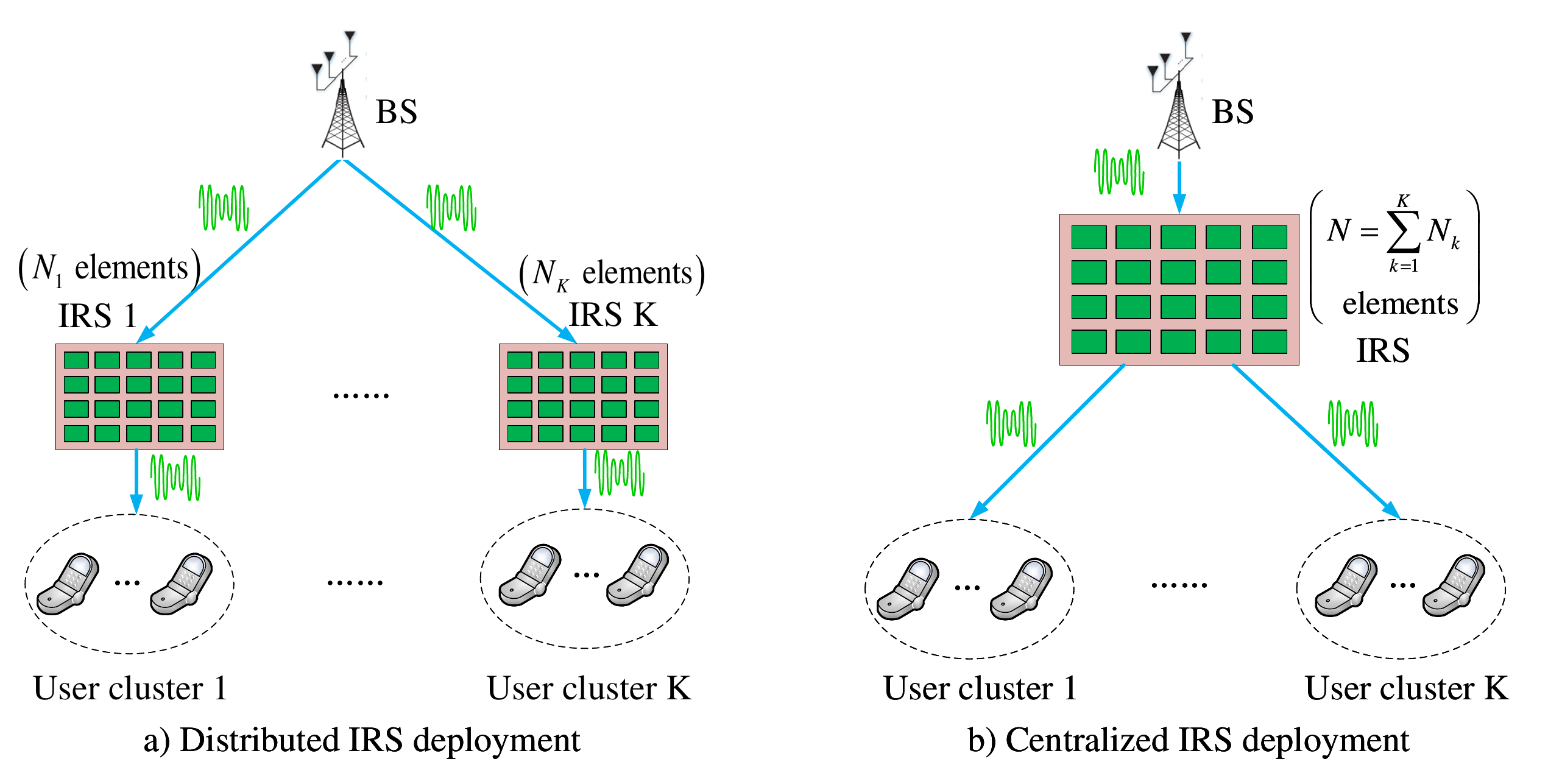}
\DeclareGraphicsExtensions.
\caption{An IRS-aided multi-user communication system with different IRS deployment strategies..}
\label{model}
\vspace{-10pt}
\end{figure}

Apart from the design of the IRS reflection pattern, IRS deployment can be flexibly optimized for further improving the system performance, thereby providing a new degree of freedom (DoF) for realizing channels customizations \cite{9745477,9740570,9423667,9261117,9976945,9427474}. However, a critical issue in IRS aided communications is that the IRS reflection link suffers from the effect of ``double fading'', which leads to severe path-loss. For the basic single-user setup, a piorneer tutorial paper \cite{wu2021intelligent} unveiled that the IRS should be placed close to either the user or the base station (BS) to minimize the resultant double path-loss effect of the IRS-aided link. For a more general multi-user setup, where different user clusters are located far apart from each other, the design of the IRS deployment should appropriately balance the performance of each individual user. In view of this issue, two typical IRS deployment architectures can be used for reducing the double path-loss of all users \cite{wu2021intelligent,9745477}, namely the \emph{distributed} and \emph{centralized IRS}, as shown in Fig. 1. For the \emph{distributed IRS} setup, all IRS elements are grouped into multiple small IRSs and each of them is placed in the vicinity of one user cluster, as illustrated in Fig. 1 (a). By contrast, all IRS elements are co-located to form a single large IRS under the \emph{centralized IRS} setup, which is placed near the BS, as illustrated in Fig. 1 (b). For these two IRS deployment architectures relying on a single-antenna BS setup, the authors of \cite{9427474} characterized the capacity regions of both the broadcast channel and multiple access channel from an information theoretical viewpoint. The analytical results of \cite{9427474} demonstrated that the \emph{centralized IRS} outperforms \emph{distributed IRS} in terms of its capacity due to the higher passive beamforming gain of the former.

Note that the result of \cite{9427474} is limited to the case with a single-antenna BS, whereas multiple antennas are generally equipped at the BS in current fifth-generation (5G) and beyond networks. In particular, the spatial domain can be fully exploited in multi-antenna networks for further improving the network capacity by serving multiple users in the same resource block simultaneously \cite{1424315}. In contrast to IRS aided single-antenna systems, the performance of IRS aided multi-antenna communication systems is determined by both the received power and the spatial multiplexing gain, which is even more important than the former in the high signal-to-noise ratio (SNR) region. Hence, the IRS deployment problem in a multi-antenna system needs to seek both increased passive beamforming gain and potential multiplexing gain. Compared to the \emph{centralized IRS}, each IRS of the distributed architecture can be flexibly deployed, which creates a rich multi-paths enviornment for improving the channel rank and thus enalbes multiple steams transmitted in parallel \cite{9976945}. By considering the fundamental tradeoff between the \emph{spatial multiplexing gain} and \emph{passive beamforming gain}, it remains to be unsolved which IRS deployment architecture achieves a higher capacity in multi-antenna systems, which thus motivates this work.

In this paper, we focus our attention on charactering the capacity of the multi-antenna broadcast channel assisted by both the \emph{distributed IRS} and \emph{centralized IRS}, as shown in Fig. 1. We aim for establishing an analytical framework to theoretically compare the capacity of the two IRS deployment architectures. To this end, the fundamental capacity limits of each IRS deployment architecture have to be characterized by capturing the distinct channel features of the two IRS deployment manners, which lay the foundation for further performance comparison. Note that the capacity characterization problem is non-trivial at all, since it involves the joint design of the IRS passive beamforming, IRS location, and BS's active beamforming. Aiming to address these issues, the main contributions of this paper are summarized as follows.

\begin{itemize}
  \item First, for drawing essential insights, we consider a special case of the line-of-sight (LoS) and homogeneous channel setup. By exploiting the unique channel structures of both IRS deployment architectures, their capacity regions are derived in closed form. For the \emph{distributed IRS} setup, an \emph{ideal IRS deployment condition} is unveiled and then we demonstrate that its capacity-achieving scheme is based on space-division multiple access (SDMA) employing maximum ratio transmission (MRT) based beamforming towards each IRS. By contrast, for the \emph{centralized IRS} setup, we reveal that its capacity-achieving scheme is based on alternating transmission among each user in a time-division multiple access (TDMA) manner by employing dynamic IRS beamforming.
  \item Second, we theoretically compare the \emph{distributed} and \emph{centralized IRS} in terms of their capacity. By carefully capturing the fundamental tradeoff between the spatial multiplexing gain and passive beamforming gain, sufficient conditions of ensuring that \emph{distributed IRS} and \emph{centralized IRS} performs better are unveiled, respectively. Our analytical results demonstrate that the sum-rate achieved by the \emph{distributed IRS} is higher than that of the \emph{centralized} one provided that $N$ is higher than a threshold, which differs from the conclusion in \cite{9427474} where the later always outperforms the former.
  \item Next, motivated by the superiority of the distributed IRS in the high $N$ regime, we focus our attention on the transmission and IRS element allocation design of the \emph{distributed IRS} architecture. By exploiting user channel correlation of intra-clusters and inter-clusters under the general Rician fading channel, we propose an efficient hybrid SDMA-TDMA multiple access scheme for harnessing both the spatial multiplexing gain and the dynamic IRS beamforming gain. Then, we investigate the issue of IRS element allocation to customize channels for both the sum-rate and the minimum user rate maximization objectives.
  \item Finally, under the general Rician fading channel, we derive the closed-form expression of the ergodic rate based on our proposed design. This performance characterization provides an efficient way for quantifying the performance erosion of the achievable sum-rate under Rician fading channels relative to that under the pure LoS channel. Extensive numerical results are presented to corroborate our theoretical findings and to unveil the benefits of the \emph{distributed IRS} in terms of improving the system capacity under various system setups.
\end{itemize}

The rest of this paper is organized as follows. Section II presents the system model of the \emph{distributed IRS} and \emph{centralized IRS} deployment architectures. In Section III, we provide a theoretical capacity comparison of these two IRS deployment architectures. Section IV addresses the transmission design and IRS element allocation problem for the \emph{distributed IRS}. Finally, we conclude in Section V.

\emph{Notations:} Boldface upper-case and lower-case  letter denote matrix and   vector, respectively.  ${\mathbb C}^ {d_1\times d_2}$ stands for the set of  complex $d_1\times d_2$  matrices. For a complex-valued vector $\bf x$, ${\left\| {\bf x} \right\|}$ represents the  Euclidean norm of $\bf x$, ${\rm arg}({\bf x})$ denotes  the phase of   $\bf x$, and ${\rm diag}(\bf x) $ denotes a diagonal matrix whose main diagonal elements are extracted from vector $\bf x$.
For a vector $\bf x$, ${\bf x}^*$ and  ${\bf x}^H$  stand for  its conjugate and  conjugate transpose respectively.   For a square matrix $\bf X$,  ${\rm{Tr}}\left( {\bf{X}} \right)$ and $\left\| {\bf{X}} \right\|_2$  respectively  stand for its trace and Euclidean norm. A circularly symmetric complex Gaussian random variable $x$ with mean $ \mu$ and variance  $ \sigma^2$ is denoted by ${x} \sim {\cal CN}\left( {{{\mu }},{{\sigma^2 }}} \right)$. ${\mathop{\rm Conv}\nolimits} \left( {\cal X} \right)$ denotes the convex hull operation of the set ${\cal X}$. $ \cup $ represents the union operation.
\vspace{-8pt}
\section{System Model}

We consider a wireless network where a multi-antenna BS serves multiple user clusters, denoted by ${{\cal A}_k}$, $k \in {\cal K} \buildrel \Delta \over = \left\{ {1, \ldots ,K} \right\}$, that are sufficiently far apart from each other. The BS is equipped with $M$ antennas and all users are equipped with a single-antenna. We focus on downlink transmission, where the BS sends independent messages to users. Moreover, a total $N$ IRS reflecting elements are deployed for enhancing wireless transmissions. For the IRS aided multi-antenna network, we consider two different deployment strategies for the $N$ available IRS elements, namely the \emph{distributed IRS} and the \emph{centralized IRS}. In particular, for the \emph{distributed IRS}, the $N$ IRS elements are grouped into $K$ distributed IRSs (see Fig. 1 (a)), where IRS $k$, $k \in {\cal K}$, with ${N_k}$ IRS elements, is deployed in the vicinity of user cluster ${{\cal A}_k}$, subject to $\sum\nolimits_{k = 1}^K {{N_k}}=N$. By contrast, for the \emph{centralized IRS}, all the available $N$ IRS elements form one single IRS, which is deployed in the vicinity of the BS (see Fig. 1 (b)). Specifically, we assume that $ L$ users are located in each user cluster ${{\cal A}_k}$, denoted by a set ${{\cal U}_k} \buildrel \Delta \over = \left\{ {u_1^k,u_2^k, \ldots ,u_{{L}}^k} \right\}$. The BS-user direct links are assumed to be severely blocked due to densely distributed obstacles. In the following, we describe the system models of both scenarios.

\subsection{Distributed IRS}
For the \emph{distributed IRS}, the baseband equivalent channels spanning from the BS to IRS $k$ and from IRS $k$ to user $u_l^k$ are denoted by ${\bf{G}}_k^{\rm{D}}\in \mathbb{C}^{{{N_k}} \times M}$ and ${\left( {{\bf{h}}_{r,kl}^{\rm{D}}} \right)^H}\in \mathbb{C}^{1 \times {{N_k}}}$, respectively. We assume that the distributed IRSs are deployed at desirable locations, so that there exists line of sight (LoS) paths to both the users and BS. Thus, we characterize the IRSs involved channels, i.e., ${\bf{G}}_k^{\rm{D}}$ and ${\left( {{\bf{h}}_{r,kl}^{\rm{D}}} \right)^H}$, by Rician fading. The BS to IRS $k$ channel can be expressed as
\begin{align}\label{channel_Gk}
{\bf{G}}_k^{\rm{D}} = \rho _{g,k}^{\rm{D}}\left( {\sqrt {\delta _k^{\rm{D}}/\left( {\delta _k^{\rm{D}} + 1} \right)} {\bf{\bar G}}_k^{\rm{D}} + \sqrt {1/\left( {\delta _k^{\rm{D}} + 1} \right)} {\bf{\tilde G}}_k^{\rm{D}}} \right),
\end{align}
where ${\left( {\rho _{g,k}^{\rm{D}}} \right)^2}$ is the large-scale path-loss, and ${\delta _k^{\rm{D}}}$ denotes the Rician factor.
The elements in ${{\bf{\tilde G}}_k^{\rm{D}}}$ are identically and independent (i.i.d.) complex Gaussian random variables with zero mean and unit variance, i.e., ${\cal C}{\cal N}\left( {0,1} \right)$. We assume that an ${N_k}\left( {{N_{v,k}} \times {N_{h,k}}} \right)$-element uniform planar array (UPA) is used at IRS $k$, $k \in {\cal K}$ and a uniform linear array (ULA) is adopted at the BS. Then, the LoS channel component ${{\bf{\bar G}}_k^{\rm{D}}}$ can be expressed as
\begin{align}\label{channel_Gk_LoS}
{{\bf{\bar G}}_k^{\rm{D}}} = {{\bf{a}}_{{\rm{S}},k}}\left( {\cos \phi _{{\rm{T}},k}^{\rm{AOA}},\sin \phi _{{\rm{T}},k}^{\rm{AOA}}\sin \theta _{{\rm{T}},k}^{\rm{AOA}}} \right){\bf{a}}_M^H\left( {\sin \theta _{{\rm{T}},k}^{\rm{AOD}}} \right),
\end{align}
where ${\theta _{{\rm{T}},k}^{\rm{AOA}}}$, ${\phi _{{\rm{T}},k}^{\rm{AOA}}}$, and ${\theta _{{\rm{T}},k}^{\rm{AOD}}}$ denote the horizontal AoA, the vertical AoA, and AoD of the BS-IRS $k$ link, respectively. Furthermore, ${{\bf{a}}_M}\left(  \cdot  \right)$ and ${{\bf{a}}_{{\rm{S}},k}}\left(  \cdot  \right)$ represent the array response vectors at the BS and IRS $k$, respectively.

Similar to the BS-IRS $k$ link, the channel spanning from IRS $k$ to user $u_l^k$ is given by
\begin{align}\label{channel_hk}
{\left( {{\bf{h}}_{r,kl}^{\rm{D}}} \right)^H} \!\!=\!\!\rho _{r,kl}^{\rm{D}}\left( {\sqrt {\frac{{\varepsilon _k^{\rm{D}}}}{{\varepsilon _k^{\rm{D}} + 1}}} {{\left( {{\bf{\bar h}}_{r,kl}^{\rm{D}}} \right)}^H} \!\!+\!\! \sqrt {\frac{1}{{\varepsilon _k^{\rm{D}} + 1}}} {{\left( {{\bf{\tilde h}}_{r,kl}^{\rm{D}}} \right)}^H}} \right),
\end{align}
where ${\left( {\rho _{r,kl}^{\rm{D}}} \right)^2}$ is the large-scale path-loss, ${\varepsilon _k^{\rm{D}}}$ denotes the Rician factor, ${{{\left( {{\bf{\tilde h}}_{r,kl}^{\rm{D}}} \right)}^H}}$ is the NLoS channel component, and ${{{\left( {{\bf{\bar h}}_{r,kl}^{\rm{D}}} \right)}^H}} = {\bf{a}}_{{\rm{S}},k}^H\left( {\cos \phi _{{\rm{R}},kl}^{\rm{AOD}},\sin \phi _{{\rm{R}},kl}^{\rm{AOD}}\sin \theta _{{\rm{R}},kl}^{\rm{AOD}}} \right)$ with ${\theta _{{\rm{R}},kl}^{\rm{AOD}}}$ and ${\phi _{{\rm{R}},kl}^{\rm{AOD}}}$ are the corresponding horizontal AoD and the vertical AoD of IRS $k$-user $u_l^k$ link. Note that the array response of the UPA can be decomposed into the Kronecker product of two ULAs as ${{\bf{a}}_{{\rm{S}},k}}\left( {X,Y} \right) = {{\bf{a}}_{{N_{v,k}}}}\left( X \right) \otimes {{\bf{a}}_{{N_{h,k}}}}\left( Y \right)$ with the array response vector of the ULA expressed by
\begin{align}\label{ULA}
{{\bf{a}}_N}\left( X \right) = \left[ {1,{e^{j\frac{{2\pi d}}{\lambda }X}}, \ldots ,{e^{j\frac{{2\pi d}}{\lambda }X\left( {N - 1} \right)}}} \right].
\end{align}
Let us denote the reflection pattern of IRS $k$ by ${\bf{\Theta }}_k^{\rm{D}} = {\mathop{\rm diag}\nolimits} \left( {{e^{j\theta _{k,1}^{\rm{D}}}}, \ldots, {e^{j\theta _{k,{N_k}}^{\rm{D}}}}} \right)\in \mathbb{C}^{{{N_k}} \times {{N_k}}}$ with $\theta _{k,n}^{\rm{D}} \in {\cal F} \buildrel \Delta \over = \left\{ {2\pi q/Q,q = 0, \ldots ,Q - 1} \right\},\forall n \in \left\{ {1, \ldots ,{N_k}} \right\}$, where $Q = {2^b}$ and $b$ denotes the number of bits adopted to quantize phase-shifts. Let ${{{\cal P}^{\rm{D}}}}$ denote the set of all possible IRS reflection patterns and thus $\left| {{{\cal P}^{\rm{D}}}} \right| = {b^{{N_k}}}$. Since the user clusters are sufficiently far apart, it is assumed that the signal reflected by IRS $j$ is negligible at the users located in user cluster ${{\cal A}_k}$, $k \ne j$. Therefore, the effective channel spanning from the BS to user $u_l^k$ is given by
\begin{align}\label{effective_channel d}
{\left( {{\bf{h}}_{kl}^{\rm{D}}\left( {{\bf{\Theta }}_k^{\rm{D}}} \right)} \right)^H} = {\left( {{\bf{h}}_{r,kl}^{\rm{D}}} \right)^H}{\bf{\Theta }}_k^{\rm{D}}{\bf{G}}_k^{\rm{D}}.
\end{align}
Let ${\bf{x}} = {\left[ {{x_1}, \ldots ,{x_M}} \right]^T}\in \mathbb{C}^{M \times 1}$ denote the vector transmitted the beamformer, where the average transmit power constraint is given by ${\mathop{\rm E}\nolimits} \left[ {\left\| {\bf{x}} \right\|_2^2} \right] \le {P_{\max }}$, where ${P_{\max }}$ denotes the maximum allowed transmitted power at the BS. Then, the vector of received symbols under the case of \emph{distributed IRS}, which is denoted by ${{\bf{y}}^{\rm{D}}} = {\left[ {y_{11}^D, \ldots ,y_{K{L}}^{\rm{D}}} \right]^T}$ (with $y_{kl}^{\rm{D}}$ representing the received signal at user $u_l^k$) is given by
\begin{align}\label{received_signal}
{{\bf{y}}^{\rm{D}}} = {{\bf{H}}^{\rm{D}}}\left( {\left\{ {{\bf{\Theta }}_k^{\rm{D}}} \right\}} \right){\bf{x}} + {\bf{z}},
\end{align}
where ${{\bf{H}}^{\rm{D}}}\left( {\left\{ {{\bf{\Theta }}_k^{\rm{D}}} \right\}} \right) \!=\! {\left[ {{\bf{h}}_{11}^{\rm{D}}\left( {{\bf{\Theta }}_k^{\rm{D}}} \right), \ldots ,{\bf{h}}_{K{L}}^{\rm{D}}\left( {{\bf{\Theta }}_k^{\rm{D}}} \right)} \right]^H}$ and ${\bf{z}} \!=\! {\left[ {{z_{11}}, \ldots ,{z_{K{L}}}} \right]^T}$ denotes the additive white Gaussian noise vector with ${z_{kl}}$ representing the noise at user $u_l^k$. Each entry in ${\bf{z}}$ is an i.i.d random variable obeying the distribution of ${\cal C}{\cal N}\left( {0,{\sigma ^2}} \right)$ with ${\sigma ^2}$ denoting the noise power.

\subsection{Centralized IRS}
For the \emph{centralized IRS} deployment, we denote the channel from the BS to the (single) IRS and from the IRS to user $u_l^k$ by ${{\bf{G}}^{\rm{C}}}\in \mathbb{C}^{N \times M}$ and ${\left( {{\bf{h}}_{r,kl}^{\rm{C}}} \right)^H}\in \mathbb{C}^{1 \times N}$, respectively. Upon adopting the Rician channel model, ${{\bf{G}}^{\rm{C}}}$ and ${\left( {{\bf{h}}_{r,kl}^{\rm{C}}} \right)^H}$ can be expressed respectively as
\begin{align}\label{channel_G_C}
{{\bf{G}}^{\rm{C}}} = \rho _g^{\rm{C}}\left( {\sqrt {\frac{{{\delta ^{\rm{C}}}}}{{{\delta ^{\rm{C}}} + 1}}} {{{\bf{\bar G}}}^{\rm{C}}} + \sqrt {\frac{1}{{{\delta ^{\rm{C}}} + 1}}} {{{\bf{\tilde G}}}^{\rm{C}}}} \right),
\end{align}
\begin{align}\label{channel_hk_C}
{\left( {{\bf{h}}_{r,kl}^{\rm{C}}} \right)^H} = \rho _{r,kl}^{\rm{C}}\left( {\sqrt {\frac{{\varepsilon _k^{\rm{C}}}}{{\varepsilon _k^{\rm{C}} + 1}}} {{\left( {{\bf{\bar h}}_{r,kl}^{\rm{C}}} \right)}^H}\sqrt {\frac{1}{{\varepsilon _k^{\rm{C}} + 1}}} {{\left( {{\bf{\tilde h}}_{r,kl}^{\rm{C}}} \right)}^H}} \right),
\end{align}
where ${\left( {\rho _g^{\rm{C}}} \right)^2}$ and ${\left( {\rho _{r,kl}^{\rm{C}}} \right)^2}$ are the large-scale path-loss, ${{\delta ^{\rm{C}}}}$ and ${\varepsilon _k^{\rm{C}}}$ are the associated Rician factors, ${{{{\bf{\tilde G}}}^{\rm{C}}}}$ and ${{{\left( {{\bf{\tilde h}}_{r,kl}^{\rm{C}}} \right)}^H}}$ are NLoS channel components whose elements are i.i.d random variables following ${\cal C}{\cal N}\left( {0,1} \right)$. With an ${N}\left( {{N_{v}} \times {N_{h}}} \right)$-element \emph{centralized IRS}, ${{{{\bf{\bar G}}}^{\rm{C}}}}$ and ${{{\left( {{\bf{\bar h}}_{r,kl}^{\rm{D}}} \right)}^H}}$ are LoS channel components, which can be expressed as ${{{\bf{\bar G}}}^{\rm{C}}} = {{\bf{a}}_{\rm{S}}}\left( {\cos \phi _{\rm{T}}^{{\rm{AOA}}},\sin \phi _{\rm{T}}^{{\rm{AOA}}}\sin \theta _{\rm{T}}^{{\rm{AOA}}}} \right){\bf{a}}_M^H\left( {\sin \theta _{\rm{T}}^{{\rm{AOD}}}} \right)$ and ${\left( {{\bf{\bar h}}_{r,kl}^{\rm{C}}} \right)^H} = {\bf{a}}_{\rm{S}}^H\left( {\cos \phi _{{\rm{RC}},kl}^{{\rm{AOD}}},\sin \phi _{{\rm{RC}},kl}^{{\rm{AOD}}}\sin \theta _{{\rm{RC}},kl}^{{\rm{AOD}}}} \right)$. Note that we have ${{\bf{a}}_{\rm{S}}}\left( {X,Y} \right) = {{\bf{a}}_{{N_v}}}\left( X \right) \otimes {{\bf{a}}_{{N_h}}}\left( Y \right)$ and $\left\{ {\phi _{\rm{T}}^{{\rm{AOA}}},\theta _{\rm{T}}^{{\rm{AOA}}},\theta _{\rm{T}}^{{\rm{AOD}}},\phi _{{\rm{RC}},kl}^{{\rm{AOD}}},\theta _{{\rm{RC}},kl}^{{\rm{AOD}}}} \right\}$ is a set of AoA/AoD information for the IRS links. Let ${{\bf{\Theta }}^{\rm{C}}} = {\rm{diag}}\left( {{e^{j\theta _1^{\rm{C}}}}, \ldots {e^{j\theta _N^{\rm{C}}}}} \right)$ denote the reflection pattern of the \emph{centralized IRS} with $\theta _n^{\rm{C}} \in {\cal F}$, $\forall n \in \left\{ {1,...,N} \right\}$ and ${{{\cal P}^{\rm{D}}}}$ denote the set of all possible IRS reflection patterns of ${{\bf{\Theta }}^{\rm{C}}}$. Hence, the effective channel vector from the BS to user $u_l^k$ under the \emph{centralized IRS} deployment can be written as
\begin{align}\label{effective_channel c}
{\left( {{\bf{h}}_{kl}^{\rm{C}}\left( {{{\bf{\Theta }}^{\rm{C}}}} \right)} \right)^H} = {\left( {{\bf{h}}_{r,kl}^{\rm{C}}} \right)^H}{{\bf{\Theta }}^{\rm{C}}}{{\bf{G}}^{\rm{C}}}.
\end{align}
Under the same expressions of the transmitted signal vector and receiver noise vector as in the \emph{distributed IRS} case, the signal received at all users can be modeled similar to \eqref{received_signal} upon replacing ${{\bf{H}}^{\rm{D}}}\left( {\left\{ {{\bf{\Theta }}_k^{\rm{D}}} \right\}} \right)$ by ${{\bf{H}}^{\rm{C}}}\left( {{{\bf{\Theta }}^{\rm{C}}}} \right)$, where ${{\bf{H}}^{\rm{C}}}\left( {{{\bf{\Theta }}^{\rm{C}}}} \right) = {\left[ {{\bf{h}}_{11}^{\rm{C}}\left( {{{\bf{\Theta }}^{\rm{C}}}} \right), \ldots ,{\bf{h}}_{K{L}}^{\rm{C}}\left( {{\bf{\Theta }}^{\rm{C}}} \right)} \right]^H}$. Comparing \eqref{effective_channel d} and \eqref{effective_channel c}, we observe that the effective channel for the BS-user $u_l^k$ link under the \emph{distributed IRS} deployment only depends on the reflection pattern ${\bf{\Theta }}_k^{\rm{D}}$ of IRS $k$, which is deployed in the vicinity of user cluster ${{\cal A}_k}$. Whereas for the \emph{centralized IRS} deployment, the effective channels of all users depend on the common IRS reflection pattern ${{\bf{\Theta }}^{\rm{C}}}$ of the single IRS.

\section{Distributed IRS Versus Centralized IRS}
In this section, we provide a theoretical performance comparison for the achievable rate under the two IRS deployment schemes. In each user cluster ${{\cal A}_k}$, we select a typical user, denoted by ${{\tilde u}^k} \in {{\cal U}_k}$, to represent the performance of its associated user cluster. For notational simplicity, we drop the user index $l$ used for the specific user in cluster ${{\cal A}_k}$. Hence, the channel from IRS $k$ (the single IRS) to user ${{\tilde u}^k}$ is denoted by ${\left( {{\bf{h}}_{r,k}^{\rm{D}}} \right)^H}$ (${\left( {{\bf{h}}_{r,k}^{\rm{C}}} \right)^H}$) and its associated large scale path-loss is ${\left( {\rho _{r,k}^{\rm{D}}} \right)^2}$ (${\left( {\rho _{r,k}^{\rm{C}}} \right)^2}$). Since we have ${\mathop{\rm rank}\nolimits} \left( {{{\bf{H}}^{\rm{D}}}\left( {\left\{ {{\bf{\Theta }}_k^{\rm{D}}} \right\}} \right)} \right) \le \min \left( {M,K} \right)$ (${\mathop{\rm rank}\nolimits} \left( {{{\bf{H}}^{\rm{C}}}\left( {{{\bf{\Theta }}^{\rm{D}}}} \right)} \right) \le \min \left( {M,K} \right)$), it is assumed that $K \le M$ for all subsequent discussions in this section. For fair comparison of the two deployment architectures, we consider the following \emph{homogeneous channel} setup, as described in \emph{Assumption 1} below.

\emph{Assumption 1} (\emph{Homogeneous Channel}): For the channel statistical properties of ${\bf{G}}_k^{\rm{D}}$, ${\left( {{\bf{h}}_{r,k}^{\rm{D}}} \right)^H}$, ${{\bf{G}}^{\rm{C}}}$, and ${\left( {{\bf{h}}_{r,k}^{\rm{C}}} \right)^H}$, it is assumed that
\begin{align}\label{assumption1}
{{{\left( {\rho _g^{\rm{C}}} \right)}^2}{{\left( {\rho _{r,1}^{\rm{C}}} \right)}^2} = {{\left( {\rho _{g,1}^{\rm{D}}} \right)}^2}{{\left( {\rho _{r,1}^{\rm{D}}} \right)}^2} =  \ldots  = {{\left( {\rho _g^{\rm{C}}} \right)}^2}{{\left( {\rho _{r,K}^{\rm{C}}} \right)}^2} = {{\left( {\rho _{g,K}^{\rm{D}}} \right)}^2}{{\left( {\rho _{r,K}^{\rm{D}}} \right)}^2}.}
\end{align}

The above homogeneous channel assumption holds in practice provided that the concatenated twin-hop path-loss factors of the IRS channels in the \emph{distributed} and \emph{centralized IRS} are the same. Based on \emph{Assumption 1}, we compare the maximum achievable rate for the two IRS deployment schemes, as detailed below.
\vspace{-8pt}
\subsection{Theoretical Performance Comparison}
We first consider the LoS channel case, i.e., $\delta _k^{\rm{D}} = \varepsilon _k^{\rm{D}} = {\delta ^{\rm{C}}} = \varepsilon _k^{\rm{C}} \to \infty $, where the IRS involved channels under the two IRS deployment schemes reduce to ${{\bf{G}}^{\rm{C}}} = \rho _g^{\rm{C}}{{{\bf{\bar G}}}^{\rm{C}}}$, ${\left( {{\bf{h}}_{r,k}^{\rm{C}}} \right)^H} = \rho _{r,k}^{\rm{C}}{\left( {{\bf{\bar h}}_{r,k}^{\rm{C}}} \right)^H}$, ${\bf{G}}_k^{\rm{D}} = \rho _{g,k}^{\rm{D}}{\bf{\bar G}}_k^{\rm{D}}$, ${\left( {{\bf{h}}_{r,k}^{\rm{D}}} \right)^H} = \rho _{r,k}^{\rm{D}}{\left( {{\bf{\bar h}}_{r,k}^{\rm{D}}} \right)^H}$.  For ease of exposition, we use ${{{\bf{a}}_{{\rm{S}},k}}}$ and ${{{\bf{a}}_{{\rm{S}}}}}$ to replace ${{{\bf{a}}_{{\rm{S}},k}}\left( {\cos \phi _{{\rm{T}},k}^{{\rm{AOA}}},\sin \phi _{{\rm{T}},k}^{{\rm{AOA}}}\sin \theta _{{\rm{T}},k}^{{\rm{AOA}}}} \right)}$ and ${{\bf{a}}_s}\left( {\cos \phi _{{\rm{RC}},kl}^{{\rm{AOD}}},\sin \phi _{{\rm{RC}},kl}^{{\rm{AOD}}}\sin \theta _{{\rm{RC}},kl}^{{\rm{AOD}}}} \right)$, respectively. We next derive the maximum achievable rate under the distributed and centralized IRS deployment, respectively.

\subsubsection{Capacity Characterization for Distributed IRS}
For the distributed IRS, the achievable rate tuple is denoted by ${{\bf{r}}^{\rm{D}}} = {\left[ {r_1^{\rm{D}}, \ldots ,r_K^{\rm{D}}} \right]^T}$ with ${r_k^{\rm{D}}}$ representing the achievable rate of user ${{\tilde u}^k}$. The active beamforming vector at the BS for user ${{\tilde u}^k}$ is denoted by ${{{\bf{w}}_k}}\in \mathbb{C}^{M \times 1}$. It is known that for a general non-degraded broadcast channel, its corresponding capacity achieving scheme is DPC \cite{1424315,10159991}. By using DPC, the capacity region along with given IRS reflection patterns and active beamforming vectors, i.e., $\left\{ {{\bf{\Theta }}_k^{\rm{D}},{{\bf{w}}_k}} \right\}$ is the region consisting of all rate-tuples that satisfy the following constraints \cite{1424315,10159991}:
\begin{align}\label{rate_distributed}
0 \le r_k^{\rm{D}} \le R_k^{\rm{D}}\left( {\left\{ {{{\bf{w}}_k}} \right\},{\bf{\Theta }}_k^{\rm{D}}} \right),\forall k,
\end{align}
with $\sum\nolimits_{k = 1}^K {{{\left\| {{{\bf{w}}_k}} \right\|}^2} \le {P_{\max }}}$, where
\begin{align}\label{defination1}
R_k^{\rm{D}}\left( {\left\{ {{{\bf{w}}_k}} \right\},{\bf{\Theta }}_k^{\rm{D}}} \right) \buildrel \Delta \over = {\log _2}\left( {1 + \frac{{{{\left| {{{\left( {{\bf{h}}_k^{\rm{D}}\left( {{\bf{\Theta }}_k^{\rm{D}}} \right)} \right)}^H}{{\bf{w}}_k}} \right|}^2}}}{{\sum\nolimits_{i = k + 1}^K {{{\left| {{{\left( {{\bf{h}}_k^{\rm{D}}\left( {{\bf{\Theta }}_k^{\rm{D}}} \right)} \right)}^H}{{\bf{w}}_i}} \right|}^2} + {\sigma ^2}} }}} \right),\forall k.
\end{align}
We denote the set characterized by \eqref{rate_distributed} as ${{\cal C}^{\rm{D}}}\left( {\left\{ {{\bf{\Theta }}_k^{\rm{D}}} \right\},\left\{ {{{\bf{w}}_k}} \right\}} \right)$. By flexibly designing $\left\{ {{\bf{\Theta }}_k^{\rm{D}},{{\bf{w}}_k}} \right\}$, any rate tuple within the union set ${{\cal C}^{\rm{D}}}\left( {\left\{ {{\bf{\Theta }}_k^{\rm{D}}} \right\},\left\{ {{{\bf{w}}_k}} \right\}} \right)$ over all feasible $\left\{ {{\bf{\Theta }}_k^{\rm{D}},{{\bf{w}}_k}} \right\}$ can be achieved. By further employing time sharing among different $\left\{ {{\bf{\Theta }}_k^{\rm{D}},{{\bf{w}}_k}} \right\}$, the capacity region of the IRS aided broadcast channel under the \emph{distributed IRS} deployment is defined as \cite{10159991}
\begin{align}\label{capacity_region_distributed}
{{\cal C}^{\rm{D}}}\!\! \buildrel \Delta \over =\!\! {\rm{Conv}}\left( {\bigcup\limits_{\theta _{k,n}^{\rm{D}} \in {\cal F},\sum\nolimits_{k = 1}^K {{{\left\| {{{\bf{w}}_k}} \right\|}^2} \le {P_{\max }}} } {{{\cal C}^{\rm{D}}}\left( {\left\{ {{\bf{\Theta }}_k^{\rm{D}}} \right\},\left\{ {{{\bf{w}}_k}} \right\}} \right)} } \right).
\end{align}

By assuming that ${N_k} = N/K$, we derive ${{\cal C}^{\rm{D}}}$ in closed form by
exploiting the special channel structure under the \emph{ideal deployment} scenario, which is provided in the following proposition.
\begin{pos}
Under the condition that
\begin{align}\label{deployment_condition}
\left| {\sin \theta _{{\rm{T}},k}^{{\rm{AOD}}} - \sin \theta _{{\rm{T}},i}^{{\rm{AOD}}}} \right| = \frac{{\lambda m}}{{dM}},\forall k \ne i,m \in \left\{ {1, \ldots ,M} \right\},
\end{align}
${{\cal C}^{\rm{D}}}$ is given by
\begin{align}\label{C_D_closed}
{{\cal C}^{\rm{D}}} = \left\{ {{{\bf{r}}^{\rm{D}}}:0 \le r_k^{\rm{D}} \le \bar r_k^{\rm{D}}} \right\},
\end{align}
under the constraint of $\sum\nolimits_{k = 1}^K {{p_k} = {P_{\max }}} $, where
\begin{align}\label{defination2}
\bar r_k^{\rm{D}} \!\!=\!\! {\log _2}\left( {1 \!\!+\!\! \frac{{{p_k}M{N^2}{{\left( {\rho _{g,k}^{\rm{D}}} \right)}^2}{{\left( {\rho _{r,k}^{\rm{D}}} \right)}^2}}}{{{K^2}{\sigma ^2}}}{{\left( {\frac{{{2^b}}}{\pi }\sin \frac{\pi }{{{2^b}}}} \right)}^2}} \right).
\end{align}
Accordingly, the maximum sum-rate of the $K$ users is obtained as
\begin{align}\label{sumrate_D_closed}
R_s^{\rm{D}} = K{\log _2}\left( {1 + \frac{{{P_{\max }}M{N^2}{{\left( {\rho _{g,1}^{\rm{D}}\rho _{r,1}^{\rm{D}}} \right)}^2}}}{{{K^3}{\sigma ^2}}}{{\left( {\frac{{{2^b}}}{\pi }\sin \frac{\pi }{{{2^b}}}} \right)}^2}} \right),
\end{align}
which is achieved by
\begin{align}\label{optimal_beamforming}
\begin{array}{l}
{\left( {{\bf{\Theta }}_k^{\rm{D}}} \right)^*} = \mathop {\arg \max }\limits_{{\bf{\Theta }}_k^{\rm{D}} \in {P^{\rm{D}}}} {\left| {{{\left( {{\bf{h}}_{r,k}^{\rm{D}}} \right)}^H}{\bf{\Theta }}_k^{\rm{D}}{{\bf{a}}_{{\rm{S}},k}}} \right|^2},\\
{\bf{w}}_k^* = \frac{{\sqrt {{P_{\max }}} {{\bf{a}}_M}\left( {\sin \theta _{{\rm{T}},k}^{{\rm{AOD}}}} \right)}}{{\sqrt {KM} }}.
\end{array}
\end{align}
\end{pos}
\begin{proof}
Please refer to Appendix A.
\end{proof}

From Proposition 1, the capacity-achieving transmission scheme under the \emph{distributed IRS} is based on SDMA by employing MRT beamforming towards each IRS. All the points on the boundary of ${\cal C}^{\rm{D}}$ can be achieved by flexibly adjusting the power allocation ${\left\{ {{p_k}} \right\}}$ under the constraint of $\sum\nolimits_{k = 1}^K {{p_k} = {P_{\max }}}$. Thanks to the deployment principle unveiled in \eqref{deployment_condition}, which is referred as \emph{ideal IRS deployment condition}, the transmit beamforming at the BS is able to simultaneously maximize the received power and fully null the inter-user interference. Besides, the role of the reflection pattern of each distributed IRS is to maximize the received power of the typical user in its user cluster. Hence, no sophisticated DPC and time sharing operation are needed due to \eqref{deployment_condition}.


\subsubsection{Capacity Characterization for Centralized IRS}
For the \emph{centralized IRS}, the achievable rate tuple is denoted by ${{\bf{r}}^{\rm{C}}} = {\left[ {r_1^{\rm{C}}, \ldots ,r_K^{\rm{C}}} \right]^T}$ with ${r_k^{\rm{C}}}$ representing the achievable rate of user ${{\tilde u}^k}$. Similar to the case of \emph{distributed IRS} deployment, the capacity region of the \emph{centralized IRS} deployment is defined as
\begin{align}\label{capacity_region_centralized}
{{\cal C}^{\rm{C}}} \buildrel \Delta \over = {\rm{Conv}}\left( {\bigcup\limits_{\theta _n^{\rm{C}} \in {\cal F},\sum\nolimits_{k = 1}^K {{{\left\| {{{\bf{w}}_k}} \right\|}^2} \le {P_{\max }}} } {{{\cal C}^{\rm{C}}}\left( {{{\bf{\Theta }}^{\rm{C}}},\left\{ {{{\bf{w}}_k}} \right\}} \right)} } \right),
\end{align}
where ${{{\cal C}^{\rm{C}}}\left( {{{\bf{\Theta }}^{\rm{C}}},\left\{ {{{\bf{w}}_k}} \right\}} \right)}$ is a set of rate-tuples satisfying the following constraints
\begin{align}\label{rate_centralized}
0 \le r_k^{\rm{C}} \le R_k^{\rm{C}}\left( {\left\{ {{{\bf{w}}_k}} \right\},{{\bf{\Theta }}^{\rm{C}}}} \right)
 \buildrel \Delta \over = {\log _2}\left( {1 + \frac{{{{\left| {{{\left( {{\bf{h}}_k^{\rm{C}}\left( {{{\bf{\Theta }}^{\rm{C}}}} \right)} \right)}^H}{{\bf{w}}_k}} \right|}^2}}}{{\sum\nolimits_{i = k + 1}^K {{{\left| {{{\left( {{\bf{h}}_k^{\rm{C}}\left( {{{\bf{\Theta }}^{\rm{C}}}} \right)} \right)}^H}{{\bf{w}}_i}} \right|}^2} + {\sigma ^2}} }}} \right),
\end{align}
with $\sum\nolimits_{k = 1}^K {{{\left\| {{{\bf{w}}_k}} \right\|}^2} \le {P_{\max }}}$. Then, we derive ${{\cal C}^{\rm{C}}}$ in closed form, as detailed in the following proposition.
\begin{pos}
As $N \to \infty $, the capacity region of the \emph{centralized IRS} deployment is given by
\begin{align}\label{C_C_closed}
{{\cal C}^{\rm{C}}} = \bigcup\limits_{{\rho _k} \in \left[ {0,1} \right],\sum\nolimits_{k = 1}^K {{\rho _k} = 1} } {\left\{ {{{\bf{r}}^C}:0 \le r_k^{\rm{C}} \le \bar r_k^C} \right\}},
\end{align}
where
\begin{align}\label{defination4}
\bar r_k^{\rm{C}} = {\rho _k}{\log _2}\left( {1 + \frac{{{P_{\max }}M{N^2}{{\left( {\frac{{{2^b}}}{\pi }\sin \frac{\pi }{{{2^b}}}} \right)}^2}}}{{{\sigma ^2}{{\left( {\rho _g^{\rm{C}}\rho _{r,k}^{\rm{C}}} \right)}^{ - 2}}}}} \right).
\end{align}
${{\cal C}^{\rm{C}}}$ is achieved by time sharing among ${\Gamma _k}$'s, which are given by
\begin{align}\label{pattern_k}
{\Gamma _k} = \left\{ {{{\bf{w}}_k} = \sqrt {\frac{{{P_{\max }}}}{M}} {{\bf{a}}_M}\left( {\sin \theta _{\rm{T}}^{{\rm{AOD}}}} \right),{{\bf{w}}_i} = {\bf{0}},\forall i \ne k,{{\left( {{{\bf{\Theta }}^{\rm{C}}}} \right)}^*} = \mathop {\arg \max }\limits_{{{\bf{\Theta }}^{\rm{C}}} \in {{\cal P}^{\rm{C}}}} {{\left| {{{\left( {{\bf{\bar h}}_{r,k}^{\rm{C}}} \right)}^H}{{\bf{\Theta }}^{\rm{C}}}{{\bf{a}}_{\rm{S}}}} \right|}^2}} \right\}.
\end{align}
Its corresponding sum-rate is
\begin{align}\label{sumrate_centralize}
R_s^{\rm{C}} = {\log _2}\left( {1 + \frac{{{P_{\max }}M{N^2}{{\left( {\rho _g^{\rm{C}}} \right)}^2}{{\left( {\rho _{r,k}^{\rm{C}}} \right)}^2}}}{{{\sigma ^2}}}{{\left( {\frac{{{2^b}}}{\pi }\sin \frac{\pi }{{{2^b}}}} \right)}^2}} \right).
\end{align}
\end{pos}
\begin{proof}
Please refer to Appendix B.
\end{proof}

Proposition 2 unveils that the capacity-achieving transmission scheme under the centralized IRS deployment is alternating transmission among each user in TDMA manner, where each user's effective channel power gain is maximized by dynamically configuring the IRS reflection pattern. Due to the assumption of the homogenous channels, each user shares the same received SNR under the optimal BS beamforming vector and IRS reflection pattern. Hence, the general superposition coding based NOMA is not needed for achieving the boundary of the capacity region.

Note that \eqref{sumrate_D_closed} and \eqref{sumrate_centralize} presented in Proposition 1 and Proposition 2 lay the theoretical foundation for the capacity comparison between the \emph{distributed} and \emph{centralized IRS} deployment.

\subsubsection{Distributed IRS versus Centralized IRS}
It is observed from \eqref{sumrate_D_closed} and \eqref{sumrate_centralize} that the passive beamforming gain achieved by the \emph{centralized IRS} is higher than that of the \emph{distributed IRS}, i.e., ${N^2}{\left( {\frac{{{2^b}}}{\pi }\sin \frac{\pi }{{{2^b}}}} \right)^2} \ge {\left( {\frac{N}{K}} \right)^2}{\left( {\frac{{{2^b}}}{\pi }\sin \frac{\pi }{{{2^b}}}} \right)^2}$ for $1 \le K \le M$. From the perspective of DoF, which determines the spatial multiplexing gain, we have
\begin{align}\label{DoF_comparison}
{d^{\rm{D}}}\!\! \buildrel \Delta \over =\!\! \mathop {\lim }\limits_{{P_{\max }} \to \infty } \frac{{R_s^D}}{{{{\log }_2}{P_{\max }}}} \!\! =\!\!  K \ge 1 = {d^{\rm{C}}} \buildrel \Delta \over = \mathop {\lim }\limits_{{P_{\max }} \to \infty } \frac{{R_s^C}}{{{{\log }_2}{P_{\max }}}},
\end{align}
with ${d^{\rm{D}}}$ and ${d^{\rm{C}}}$ respectively denoting the DoF of \emph{distributed IRS} and \emph{centralized IRS}, which indicates that the multiplexing gain achieved by the former is higher than that of the latter. By taking both the passive beamforming gain and multiplexing gain into account, the comparison outcome for these two cases depends on the specific system parameters. First, we unveil sufficient conditions for ensuring that \emph{distributed IRS} outperforms \emph{centralized IRS} and those of its opposite in the following theorem.
\begin{thm}
For $\forall K \in \left\{ {1, \ldots ,M} \right\}$, we have $R_s^{\rm{D}} \le R_s^{\rm{C}}$ provided that
\begin{align}\label{sufficient_condition1}
N \le \sqrt {\frac{{{C_{{\rm{th}}}}{\sigma ^2}}}{{{P_{\max }}M{{\left( {\rho _{g,1}^{\rm{D}}} \right)}^2}{{\left( {\rho _{r,1}^{\rm{D}}} \right)}^2}{{\left( {\frac{{{2^b}}}{\pi }\sin \frac{\pi }{{{2^b}}}} \right)}^2}}}} ,
\end{align}
where ${{C_{{\rm{th}}}}}$ is the unique solution of the equation
\begin{align}\label{equation}
g\left( x \right) \buildrel \Delta \over = \ln \left( {1 + x} \right) - 3 + \frac{3}{{1 + x}} = 0
\end{align}
located in $\left( {0,\infty } \right)$. Otherwise, $R_s^{\rm{D}} \ge R_s^{\rm{C}}$ for $\forall K \in \left\{ {1, \ldots ,M} \right\}$ if
\begin{align}\label{sufficient_condition2}
N \ge M\sqrt {\frac{{{C_{{\rm{th}}}}{\sigma ^2}}}{{{P_{\max }}{{\left( {\rho _{g,1}^{\rm{D}}} \right)}^2}{{\left( {\rho _{r,1}^{\rm{D}}} \right)}^2}{{\left( {\frac{{{2^b}}}{\pi }\sin \frac{\pi }{{{2^b}}}} \right)}^2}}}}.
\end{align}
\end{thm}
\begin{proof}
Please refer to Appendix C.
\end{proof}

%

In Theorem 1, \eqref{sufficient_condition1} and \eqref{sufficient_condition2} serve as sufficient conditions for ensuring that \emph{centralized IRS} outperforms \emph{distributed IRS} and its opposite case, respectively. It is observed that the \emph{distributed IRS} deployment is preferable provided that the total number of IRS elements is sufficiently large. The reason is that the multiplexing gain of \emph{distributed IRS} is higher than that of \emph{centralized IRS}, which leads a faster increase of sum-rate with the receive SNR in the former case. Increasing the number of IRS elements helps enhance the SNR at users, which is beneficial for significantly improving the sum-rate under the \emph{distributed IRS} deployment. By contrast, when the number of IRS elements is small, the receive SNR at users is low and thus the sum-rate is mainly limited by the passive beamforming gain achieved, rather than by the multiplexing gain. Hence, \emph{centralized IRS} is preferable in this case due to the higher passive beamforming gain.

To gain more useful insights, we focus on the asymptotically high SNR case with $\left( {{P_{\max }}/{\sigma ^2}} \right) \to \infty$, such as ${P_{\max }}{\left( {\rho _{g,1}^{\rm{D}}} \right)^2}{\left( {\rho _{r,1}^{\rm{D}}} \right)^2}/{\sigma ^2} \gg 1$. In this case, the sufficient and necessary condition for ensuring that \emph{distributed IRS} outperforms \emph{centralized IRS} is unveiled in the following theorem.
\begin{thm}
Under the assumption of ${P_{\max }}{\left( {\rho _{g,1}^{\rm{D}}} \right)^2}{\left( {\rho _{r,1}^{\rm{D}}} \right)^2}/{\sigma ^2} \gg 1$, we have $R_s^{\rm{D}} \ge R_s^{\rm{C}}$ for any given $K$ satisfying $2 \le K \le M$ if and only if
\begin{align}\label{n_s_condition}
N \ge {N_{{\rm{th}}}} \buildrel \Delta \over = \sqrt {\frac{{{\sigma ^2}}}{{{P_{\max }}M{{\left( {\rho _{g,1}^{\rm{D}}} \right)}^2}{{\left( {\rho _{r,1}^{\rm{D}}} \right)}^2}{{\left( {\frac{{{2^b}}}{\pi }\sin \frac{\pi }{{{2^b}}}} \right)}^2}}}} {K^{\frac{{3K}}{{2\left( {K - 1} \right)}}}}.
\end{align}
\end{thm}
\begin{proof}
Based on the assumption of ${P_{\max }}{\left( {\rho _{g,1}^{\rm{D}}} \right)^2}{\left( {\rho _{r,1}^{\rm{D}}} \right)^2}/{\sigma ^2} \gg 1$, we have
\begin{align}\label{approximation_distributed}
R_s^{\rm{D}}\mathop  \to \limits^{a.s.} K{\log _2}\left( {\frac{{{{\tilde \gamma }_0}{N^2}}}{{{K^3}}}} \right)= K\left( {2{{\log }_2}N + {{\log }_2}{{\tilde \gamma }_0} - 3{{\log }_2}K} \right) \buildrel \Delta \over = \tilde R_s^{\rm{D}}
\end{align}
\vspace{-8pt}
\begin{align}\label{approximation_centralized}
R_s^{\rm{C}}\mathop  \to \limits^{a.s.} {\log _2}\left( {{{\tilde \gamma }_0}{N^2}} \right) = 2{\log _2}N + {\log _2}{{\tilde \gamma }_0} \buildrel \Delta \over = \tilde R_s^{\rm{C}},
\end{align}
with ${{\tilde \gamma }_0} = {P_{\max }}M{\left( {\rho _{g,1}^{\rm{D}}} \right)^2}{\left( {\rho _{r,1}^{\rm{D}}} \right)^2}{\left( {\frac{{{2^b}}}{\pi }\sin \frac{\pi }{{{2^b}}}} \right)^2}/{\sigma ^2}$. By solving the inequality $\tilde R_s^{\rm{D}} \ge \tilde R_s^{\rm{C}}$ based on \eqref{approximation_distributed} and \eqref{approximation_centralized}, \eqref{n_s_condition} is naturally obtained.
\end{proof}

\begin{rem}
$\left\lceil {{N_{{\rm{th}}}}} \right\rceil $ defined in Theorem 2 represents the minimum number of total IRS elements required for \emph{distributed IRS} to outperforme \emph{centralized IRS}, which is a function of system parameters, e.g., $M$, $b$, ${P_{\max }}$, and $K$. It is seen from \eqref{n_s_condition} that ${{N_{{\rm{th}}}}}$  monotonously decreases with $M$, $b$, and ${P_{\max }}$, which suggests that the practical operating region for \emph{distributed IRS} can be extended by increasing the transmit power, the number  of antennas or quantization bits for IRS phase-shifts, since it is helpful for improving the received power. Upon relaxing ${{N_{{\rm{th}}}}}$ as a continuous variable, we obtain
\begin{align}\label{Elements_vs_K}
\frac{{\partial \ln {N_{{\rm{th}}}}}}{{\partial  K}} = \frac{{K - 1 - \ln K}}{{{{\left( {K - 1} \right)}^2}}}\mathop  \ge \limits^{\left( a \right)} 0,
\end{align}
where (a) holds due to $\ln x \le x - 1$ for $\forall x \ge 1$. Hence, $\left\lceil {{N_{{\rm{th}}}}} \right\rceil $ monotonically increases with $K$, which indicates that in total more IRS elements are needed for \emph{distributed IRS} to outperform \emph{centralized IRS}, when the number of user clusters is high.
\end{rem}
\begin{rem}
In the large-IRS regime, i.e., ${N \to \infty }$, we have
\begin{align}\label{distributed_versus_centralized_large_N}
\mathop {\lim }\limits_{N \to \infty } \frac{{R_s^{\rm{D}}}}{{R_s^{\rm{C}}}} = \mathop {\lim }\limits_{N \to \infty } \frac{{K\left( {2{{\log }_2}N + {{\log }_2}{{\tilde \gamma }_0} - 3{{\log }_2}K} \right)}}{{2{{\log }_2}N + {{\log }_2}{{\tilde \gamma }_0}}} = K,
\end{align}
which further demonstrates that \emph{distributed IRS} is more appealing in practical systems, provided that a high total number of IRS elements is affordable.
\end{rem}

\subsection{Numerical Results}
In this subsection, we provide numerical results to verify our theoritical findings under the LoS and homogeneous channel setup, as shown in \emph{Assumption 1}. We set $M=5$, $K = 4$.  ${P_{\max }} = 30$ dBm,  and $ {\sigma ^2}=- 90$ dBm. The distributed IRSs are deployed according to the \emph{ideal deployment condition} unveiled in \eqref{deployment_condition}. For the \emph{homogeneous channel} setup, the two-hop path-loss via the IRS link is set as ${\left( {\rho _g^{\rm{C}}\rho _{r,k}^{\rm{C}}} \right)^2} = {\left( {\rho _{g,k}^{\rm{D}}\rho _{r,k}^{\rm{D}}} \right)^2} =  - 140$ dB, $\forall k$ and the number of IRS elements for each IRS under the distributed IRS architecture is set to ${N_k} = N/K,\forall k$.

For comparison, we consider the following schemes: 1) \textbf{Distributed IRS-optimal}: the optimal SDMA-based transmission scheme unveiled in Proposition 1 is employed under the \emph{distributed IRS}; 2) \textbf{Centralized IRS-optimal}: the optimal TDMA-based transmission scheme unveiled in Proposition 2 is employed under the \emph{centralized IRS}; 3) \textbf{Distributed IRS-TDMA}: Under the \emph{distributed IRS}, the TDMA scheme is adopted. In Fig. \ref{sumrate_power}, we plot the sum-rate of all the schemes considered versus the maximum transmit power at the BS. It is observed that the sum-rate of the \emph{distributed IRS} under the optimal transmission scheme increases more sharply with the power than that of the \emph{centralized IRS}. This is expected since the \emph{distributed IRS} enjoys a higher spatial multiplexing gain, which agrees with our analysis in \eqref{DoF_comparison}. As such, the sum-rate of the \emph{distributed IRS} gradually exceeds that of the \emph{distributed IRS} as ${P_{\max }}$ increases and the relative performance gain becomes more pronounced for a high ${P_{\max }}$. Additionally, the \emph{centralized IRS} always outperforms the \emph{distributed IRS} under the TDMA scheme. This is due to fact that each user is only covered by its local IRS under the distributed IRS deployment, which results in a lower passive beamforming gain compared to the centralized IRS. The results highlight the importance of employing the most appropriate transmission scheme for each IRS deployment architecture.

\begin{figure}
\begin{minipage}[t]{0.4\linewidth}
\centering
\includegraphics[width=2.9in]{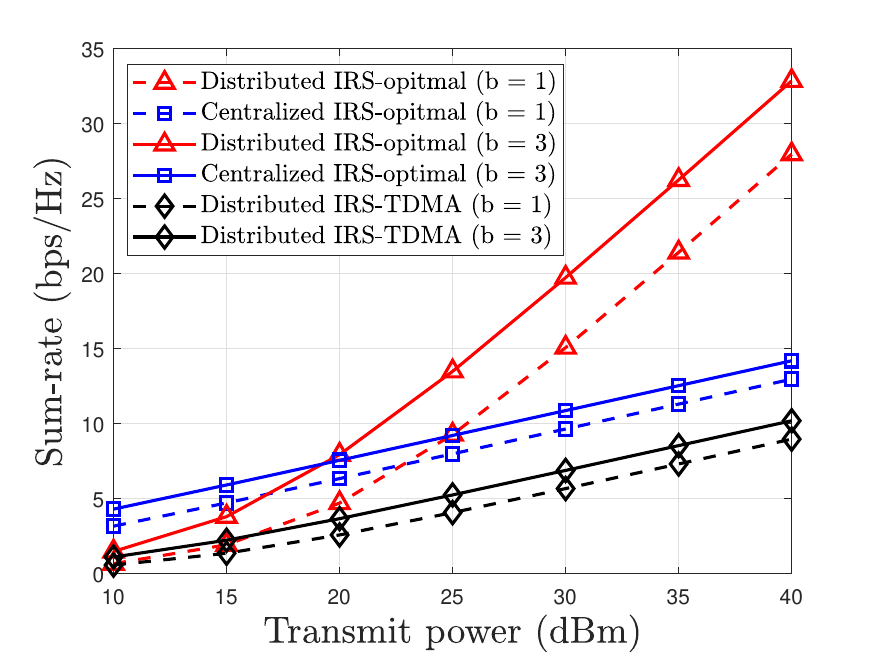}
\caption{Sum-rate versus ${P_{\max }}$ with $N=200$.}
\label{sumrate_power}
\end{minipage}%
\hfill
\begin{minipage}[t]{0.4\linewidth}
\centering
\includegraphics[width=2.9in]{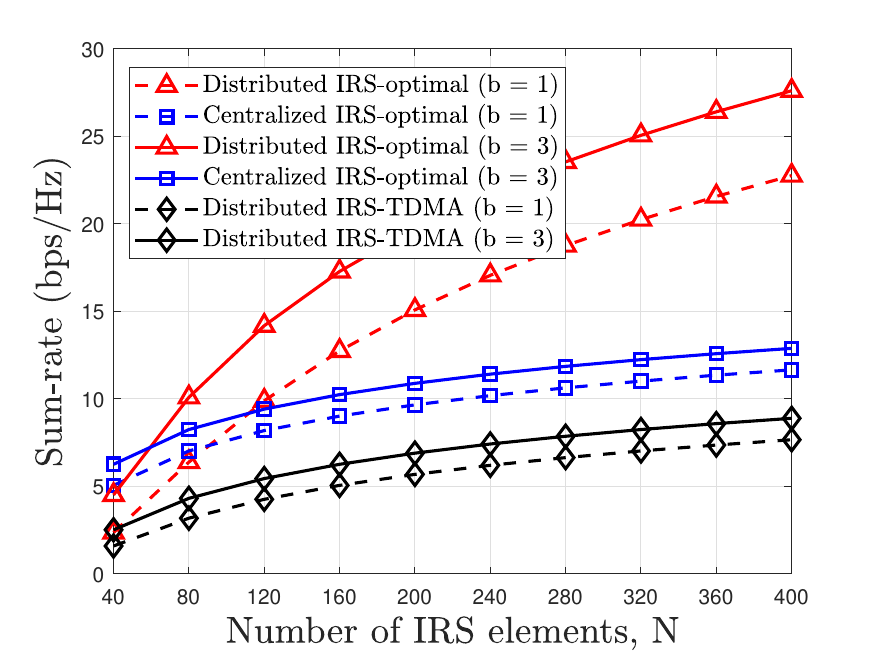}
\caption{Sum-rate versus $N$.}
\label{sumrate_elements}
\end{minipage}
\vspace{-16pt}
\end{figure}

In Fig. \ref{sumrate_elements}, we show the sum-rate versus $N$. It is observed that the \emph{centralized IRS} outperforms \emph{distributed IRS} in the low-$N$ regime. In the low-$N$ regime, the sum-rate is mainly restricted by the  power received at the user. Compared to \emph{distributed IRS}, \emph{centralized IRS} enjoys the advantages of reaping higher passive beamforming gain, which is beneficial for substantially improving the received power. Nevertheless, the \emph{distributed IRS} gradually outperforms \emph{centralized IRS} as $N$ increases. This is due to the fact that \emph{centralized IRS} has limited DoF for spatial multiplexing. As $N$ becomes large, the power received at the user becomes sufficient and the benefits brought about by spatial multiplexing under the distributed IRS architecture become dominant. The results demonstrate the superiority of employing distributed IRS when the total number of available IRS elements is large, which validates Theorem 1.

\begin{figure}[!t]
\setlength{\abovecaptionskip}{-5pt}
\setlength{\belowcaptionskip}{-5pt}
\centering
\includegraphics[width= 0.4\textwidth]{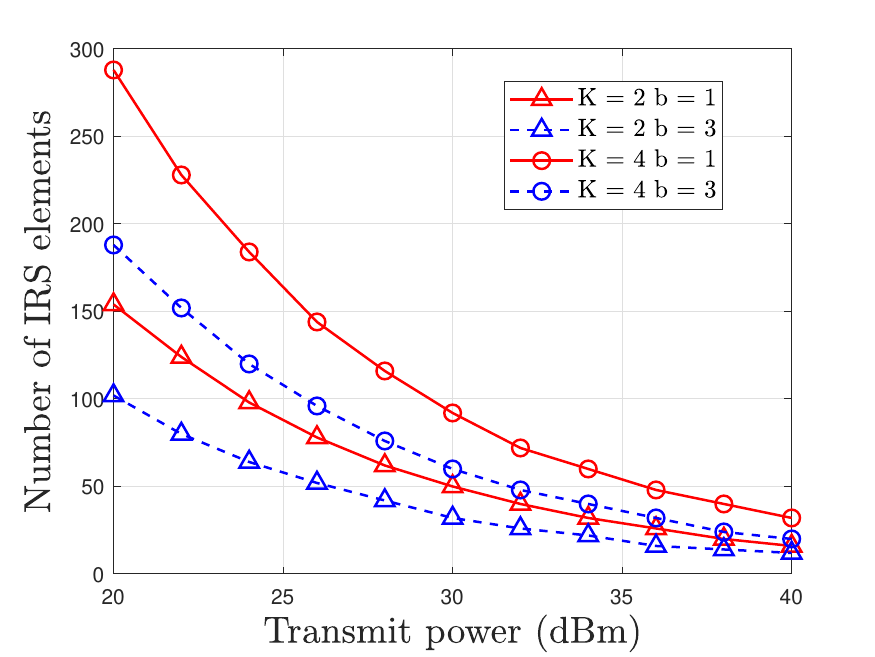}
\DeclareGraphicsExtensions.
\caption{$N$ required for ensuring distributed IRS to outperform centralized IRS.}
\label{R_elements}
\vspace{-10pt}
\end{figure}

To unveil the operating region of the distributed IRS, we quantify the total number of IRS elements required for ensuring \emph{distributed IRS} to outperform \emph{centralized IRS} in Fig. \ref{R_elements}. It is observed from Fig. \ref{R_elements} that the requirement for the total number IRS elements can be alleviated by increasing the transmit power at the BS, increasing the number of quantization bits at the IRS, and reducing the number of scheduled users, i.e., $K$. The reason is that increasing ${P_{\max }}$, $b$ or reducing $K$, is helpful for increasing the power received at the users. These results are consistent with our analysis in Remark 1.

\vspace{-8pt}
\section{Transmission and Element Allocation design for distributed IRS}
Section III has theoritically unveiled that the \emph{distributed IRS} is more appealing, when a large number of IRS elements is affordable. Motivated by the potential of \emph{distributed IRS} for achieving high capacity, we next focus on the design of the transmission scheme and IRS element allocation under the \emph{distributed IRS}. We first propose a hybrid SDMA-TDMA multiple access scheme by exploiting the user channel correlation of intra-clusters and inter-clusters. Then, the IRS element allocation problem is studied both in terms of sum-rate maximization and minimum user rate maximization.
\vspace{-8pt}
\subsection{Hybrid SDMA-TDMA Scheme}
Under the assumption of $\delta _1^{\rm{D}}{\rm{ = }} \ldots {\rm{ = }}\delta _K^{\rm{D}} = {\delta ^{\rm{D}}}$, $\varepsilon _k^{\rm{D}} \to \infty, \forall k$, we obtain the following proposition to capture the effect of channel correlation under the \emph{distributed IRS} deployment.
\begin{pos}
Under a randomly given ${\bf{\Theta }}_k^{\rm{D}}$, i.e., $\theta _{k,n}^{\rm{D}} \sim {\cal U}\left[ {0,2\pi } \right]$, $\forall k,n$, we have
\begin{align}\label{correlation1}
\rho _{kl,ml'}^2 \!\buildrel \Delta \over = \!\frac{{{\rm{E}}\left\{ {{{\left| {{{\left( {{\bf{h}}_{kl}^{\rm{D}}\left( {{\bf{\Theta }}_k^{\rm{D}}} \right)} \right)}^H}\left( {{\bf{h}}_{ml'}^{\rm{D}}\left( {{\bf{\Theta }}_m^{\rm{D}}} \right)} \right)} \right|}^2}} \right\}}}{{{\rm{E}}\left\{ {\left\| {{\bf{h}}_{kl}^{\rm{D}}\left( {{\bf{\Theta }}_k^{\rm{D}}} \right)} \right\|_2^2} \right\}{\rm{E}}\left\{ {\left\| {{\bf{h}}_{ml'}^{\rm{D}}\left( {{\bf{\Theta }}_m^{\rm{D}}} \right)} \right\|_2^2} \right\}}}{\rm{ = }}\frac{{2{\delta ^{\rm{D}}} + 1}}{{{{\left( {{\delta ^{\rm{D}}} + 1} \right)}^2}}}\frac{1}{M},\forall k \ne m,
\end{align}
\begin{align}\label{correlation2}
\rho _{kl,kl'}^2 \buildrel \Delta \over = \frac{{{\rm{E}}\left\{ {{{\left| {{{\left( {{\bf{h}}_{kl}^{\rm{D}}\left( {{\bf{\Theta }}_k^{\rm{D}}} \right)} \right)}^H}\left( {{\bf{h}}_{kl'}^{\rm{D}}\left( {{\bf{\Theta }}_k^{\rm{D}}} \right)} \right)} \right|}^2}} \right\}}}{{{\rm{E}}\left\{ {{{\left\| {{\bf{h}}_{kl}^{\rm{D}}\left( {{\bf{\Theta }}_k^{\rm{D}}} \right)} \right\|}^2}} \right\}{\rm{E}}\left\{ {{{\left\| {{\bf{h}}_{kl'}^{\rm{D}}\left( {{\bf{\Theta }}_k^{\rm{D}}} \right)} \right\|}^2}} \right\}}}{\rm{ = }}\frac{{2{\delta ^{\rm{D}}} + 1}}{{{{\left( {{\delta ^{\rm{D}}} + 1} \right)}^2}}}\frac{1}{M} + \frac{{{{\left( {{\delta ^{\rm{D}}}} \right)}^2}}}{{{{\left( {{\delta ^{\rm{D}}} + 1} \right)}^2}}},\forall l \ne l',
\end{align}
as ${N_k} \to \infty$, where $\rho _{kl,ml'}^2$ and $\rho _{kl,kl'}^2$ denote the squared-correlation coefficients of the user pairs $\left( {u_l^k,u_{l'}^m} \right)$ and $\left( {u_l^k,u_{l'}^k} \right)$, respectively.
\end{pos}
\begin{proof}
Please refer to Appendix D.
\end{proof}

It is plausible from Proposition 3 that $\rho _{kl,ml'}^2 < \rho _{kl,kl'}^2,\forall k \ne m,l \ne l'$ provided that ${\delta ^{\rm{D}}} > 0$, which explicitly demonstrates that the squared-correlation of channels for the users located in different clusters is lower than that of the users located in the same cluster. As the Rician factor ${{\delta ^{\rm{D}}}} $ increases, the difference of the two squared-correlations $\rho _{kl,kl'}^2 - \rho _{kl,ml'}^2$ increases. Note that we have $\rho _{kl,ml'}^2 = 0, \forall k \ne m$ and $\rho _{kl,kl'}^2 = 1, \forall l \ne l'$ as ${\delta ^{\rm{D}}} \to \infty $. Hence, Proposition 3 motivates us to propose an efficient hybrid SDMA-TDMA scheme to harness both the spatial multiplexing gain and dynamic IRS beamforming gain, as described below.

Recall from Section II that there are $L$ users\footnote{The proposed hybrid SDMA-TDMA scheme is also applicable to the scenario where the number of users in each cluster is different. The key idea is to schedule the intra-cluster users via the round robin scheme and to serve the inter-cluster users simultaneously via SDMA.} in each user cluster associated with IRS $k$, denoted by the set ${{\cal U}_k} \buildrel \Delta \over = \left\{ {u_1^k,u_2^k, \ldots ,u_{{L}}^k} \right\}$. We focus on a specific channel coherence block ${\cal T}$, whose time duration is denoted by $T$. First, the channel's coherence interval ${\cal T}$ is equally partitioned into $L$ orthogonal time slots (TSs), denoted by ${{\cal T}_l}$, $\forall l \in {\cal L} \buildrel \Delta \over = \left\{ {1, \ldots ,L} \right\}$, and the time duration of ${{\cal T}_l}$ is $T/L$. Then, all users are naturally divided into $L$ disjoint groups, denoted by ${{\cal G}_l} = \left\{ {u_l^1,u_l^2 \ldots ,u_l^K} \right\}$, $\forall l \in {\cal L}$. The users in group ${{\cal G}_l}$ are scheduled in TS ${{\cal T}_l}$ via SDMA. In particular, the transmit beamforming vectors at the BS and IRS reflection pattern in TS ${{\cal T}_l}$ are configured as
\begin{align}\label{joint_bf_configuration}
{\bf{\Theta }}_k^{\rm{D}}\left[ l \right] = \mathop {\arg \max }\limits_{{\bf{\Theta }}_k^{\rm{D}} \in {{\cal P}^{\rm{D}}}} {\left| {{{\left( {{\bf{\bar h}}_{r,kl}^{\rm{D}}} \right)}^H}{\bf{\Theta }}_k^{\rm{D}}{{\bf{a}}_{{\rm{S}},k}}} \right|^2},{{\bf{w}}_k}\left[ l \right] = \sqrt {{p_k}} \frac{{{{\bf{a}}_M}\left( {\sin \theta _{{\rm{T}},k}^{{\rm{AOD}}}} \right)}}{{\sqrt M }},\forall k,l,
\end{align}
with $\sum\nolimits_{k = 1}^K {{p_{kl}} = {P_{\max }}}$. Based on \eqref{joint_bf_configuration}, the sum-rate of all users in ${\cal T}$ can be written as
\begin{align}\label{sumrate_all_cluster}
\bar R = \sum\limits_{l = 1}^L {\sum\limits_{k = 1}^K {\frac{1}{L}{{\log }_2}\left( {1 + {\rm{SIN}}{{\rm{R}}_{k,l}}} \right)} },
\end{align}
where
\begin{align}\label{SINR_kl}
{\rm{SIN}}{{\rm{R}}_{k,l}} = \frac{{{{\left| {{{\left( {{\bf{h}}_{kl}^{\rm{D}}\left( {{\bf{\Theta }}_k^{\rm{D}}\left[ l \right]} \right)} \right)}^H}{{\bf{w}}_k}\left[ l \right]} \right|}^2}}}{{\sum\nolimits_{m = 1,m \ne k}^K {{{\left| {{{\left( {{\bf{h}}_{kl}^{\rm{D}}\left( {{\bf{\Theta }}_k^{\rm{D}}\left[ l \right]} \right)} \right)}^H}{{\bf{w}}_k}\left[ l \right]} \right|}^2} + {\sigma ^2}} }}
\end{align}
denotes the received signal-tointerference-plus-noise ratio (SINR) of user $u_l^k$.

\begin{rem}
It is observed from \eqref{joint_bf_configuration} that the design of $\left\{ {{{\bf{w}}_k}\left[ l \right],{\bf{\Theta }}_k^{\rm{D}}\left[ l \right]} \right\}$ relies only on the statistical CSI, i.e., the locations, AoD of the BS, AoA/AoD of the IRSs, and AoA of the users. Note that the statistical CSI varies slowly and hence remains near-constant for a long time. Accordingly, the proposed hybrid SDMA-TDMA scheme requires low channel estimation overhead and computational complexity, which is more appealing in practical systems with a high $N$.
\end{rem}
\vspace{-12pt}
\subsection{IRS Element Allocation Design}
In this subsection, we study the IRS element allocation problem under the proposed hybrid SDMA-TDMA transmission scheme. Note that the random NLoS components of the IRS involved channels cannot be applied to determine the number of IRS elements in each user cluster. Motivated by this, the IRS element allocation problem is investigated under a LoS channel scenario. In each user cluster, we select a user located at the boundary of ${{\cal A}_k}$, which represents the performance of $k$-th cluster. Without loss of generality, the selected user in the $k$-th cluster is assumed to be $u_L^k$. Upon substituting \eqref{joint_bf_configuration} into \eqref{SINR_kl} under the LoS channel setup, the rate of user $u_L^k$ can be expressed as
\begin{align}\label{rate_kL}
R_{kL}^{\rm{D}} \!\!=\!\! {\log _2}\left( {1 + \frac{{{p_{kL}}MN_k^2{{\left( {\rho _{g,k}^{\rm{D}}} \right)}^2}{{\left( {\rho _{r,kL}^{\rm{D}}} \right)}^2}}}{{{\sigma ^2}}}{{\left( {\frac{{{2^b}}}{\pi }\sin \frac{\pi }{{{2^b}}}} \right)}^2}} \right).
\end{align}
In the following, we study the sum-rate maximization and the minimum user rate maximization problems, respectively.

\subsubsection{Minimum User Rate Maximization}
For minimum user rate maximization, the corresponding optimization problem by jointly optimizing the power allocation and IRS element allocation can be formulated as follows.
\begin{subequations}\label{C13}
\begin{align}
\label{C13-a}\mathop {\max }\limits_{\left\{ {{p_k}} \right\},\left\{ {{N_k}} \right\}} \;\;&\mathop {\min }\limits_{k \in {\cal K}} \;\;R_{kL}^{\rm{D}}\\
\label{C13-b}{\rm{s.t.}}\;\;\;\;\;&\sum\nolimits_{k = 1}^K {{p_{kL}}}  = {P_{\max }},\\
\label{C13-c}&\sum\nolimits_{k = 1}^K {{N_k}}  = N, \\
\label{C13-d}&{N_k} \in {\mathbb{N}},\forall k \in {\cal K}.
\end{align}
\end{subequations}
Note that constraints \eqref{C13-b} and \eqref{C13-c} represent the transmit power constraint at the BS and the deployment budget for the total number of IRS elements, respectively. Problem \eqref{C13} is challenging to be solved optimally since ${{p_k}}$ and ${{N_k}}$ are tightly coupled in the objective function \eqref{C13-a}, which renders the design objective a complicated function. Moreover, constraint \eqref{C13-d} is non-convex, since the number of IRS elements deployed in each cluster is discrete.

To overcome the above challenges, we first relax the value of ${{N_k}}$ into a continuous value ${{\tilde N}_k}$ and then the integer rounding technique is employed to reconstruct the optimal solution of the original optimization problem, which leads to the following optimization problem:
\begin{subequations}\label{C14}
\begin{align}
\label{C14-a}\mathop {\max }\limits_{\left\{ {{p_k}} \right\},\left\{ {{{\tilde N}_k}} \right\}} \;\;&\mathop {\min }\limits_{k \in {\cal K}} \;\;R_{kL}^{\rm{D}}\\
\label{C14-b}{\rm{s.t.}}\;\;\;\;\;&\sum\nolimits_{k = 1}^K {{{\tilde N}_k}}  = N, \\
\label{C14-c}&\eqref{C13-b}.
\end{align}
\end{subequations}
Although problem \eqref{C14} is still non-convex, we obtain its optimal solution in the following proposition by exploiting its particular structure.
\begin{pos}
The optimal solution $\left\{ {p_{kL}^*,\tilde N_k^*} \right\}$ of problem \eqref{C14} is
\begin{align}\label{optimal_solution_C14}
\tilde N_k^* = \frac{{{{\left( {2{{\left( {\rho _{g,k}^{\rm{D}}\rho _{r,kL}^{\rm{D}}} \right)}^{ - 2}}} \right)}^{1/3}}N}}{{\sum\nolimits_{j = 1}^K {{{\left( {2{{\left( {\rho _{g,j}^{\rm{D}}\rho _{r,jL}^{\rm{D}}} \right)}^{ - 2}}} \right)}^{1/3}}} }},\;\;p_{kL}^* = \frac{{{P_{\max }}{{\left( {\rho _{g,k}^{\rm{D}}\rho _{r,kL}^{\rm{D}}\tilde N_k^*} \right)}^{ - 2}}}}{{\sum\nolimits_{j = 1}^K {{{\left( {\rho _{g,j}^{\rm{D}}\rho _{r,jL}^{\rm{D}}\tilde N_k^*} \right)}^{ - 2}}} }},\forall k \in K.
\end{align}
\end{pos}
\begin{proof}
First, we show that the condition
\begin{align}\label{condition_snr1}
{\Upsilon _1} = {\Upsilon _2} =  \ldots  = {\Upsilon _K} = \Upsilon
\end{align}
is satisfied at the optimal solution by using the method of contradiction, where
\begin{align}\label{snr_definition1}
{\Upsilon _k} = \frac{{{p_{kL}}M\tilde N_k^2{{\left( {\rho _{g,k}^{\rm{D}}} \right)}^2}{{\left( {\rho _{r,kL}^{\rm{D}}} \right)}^2}}}{{{\sigma ^2}}}{\left( {\frac{{{2^b}}}{\pi }\sin \frac{\pi }{{{2^b}}}} \right)^2}
\end{align}
denotes the receive SNR at user $u_L^k$. Assume that $\tilde \Xi  = \left\{ {{{\tilde p}_{kL}},{{\tilde N}_k}} \right\}$ is the optimal solution of problem \eqref{C14}, which yields ${\Upsilon _1} =  \ldots  = {\Upsilon _{K - 1}} > {\Upsilon _K}$. Then, we construct a different solution $\mathord{\buildrel{\lower3pt\hbox{$\scriptscriptstyle\frown$}}
\over \Xi }  = \left\{ {{{\mathord{\buildrel{\lower3pt\hbox{$\scriptscriptstyle\frown$}}
\over p} }_{kL}},{{\mathord{\buildrel{\lower3pt\hbox{$\scriptscriptstyle\frown$}}
\over N} }_k}} \right\}$, where ${{\mathord{\buildrel{\lower3pt\hbox{$\scriptscriptstyle\frown$}}
\over p} }_{kL}} = {{\tilde p}_{kL}}$ for $k < K - 1$, ${{\mathord{\buildrel{\lower3pt\hbox{$\scriptscriptstyle\frown$}}
\over p} }_{\left( {K - 1} \right)L}} = {{\tilde p}_{\left( {K - 1} \right)L}} + \Delta p$, ${{\mathord{\buildrel{\lower3pt\hbox{$\scriptscriptstyle\frown$}}
\over p} }_{KL}} = {{\tilde p}_{KL}} - \Delta p$, and ${{\mathord{\buildrel{\lower3pt\hbox{$\scriptscriptstyle\frown$}}
\over N} }_k} = {{\tilde N}_k},\forall k$. Note that the value of $\Delta p$ is selected to keep ${\Upsilon _{K - 1}} = {\Upsilon _K}$. It can be readily verified that ${\mathord{\buildrel{\lower3pt\hbox{$\scriptscriptstyle\frown$}}
\over \Xi } }$ is also a feasible solution to \eqref{C14} and its achieved objective value is larger than that under the solution $\tilde \Xi  = \left\{ {{{\tilde p}_{kL}},{{\tilde N}_k}} \right\}$, which contradicts that ${\tilde \Xi }$ is optimal. Hence, \eqref{condition_snr1} must be satisfied at the optimal solution. Let
\begin{align}\label{tao_k}
{\Gamma _k} = \frac{{M{{\left( {\rho _{g,k}^{\rm{D}}} \right)}^2}{{\left( {\rho _{r,kL}^{\rm{D}}} \right)}^2}}}{{{\sigma ^2}}}{\left( {\frac{{{2^b}}}{\pi }\sin \frac{\pi }{{{2^b}}}} \right)^2}, \forall k.
\end{align}
Then, the optimal $\left\{ {{p_{kL}}} \right\}$ under the arbitrarily given $\left\{ {{{\tilde N}_k}} \right\}$ is ${p_{kL}} = \Upsilon /\left( {{\Gamma _k}\tilde N_k^2} \right)$.  Upon substituting ${p_{kL}} = \Upsilon /\left( {{\Gamma _k}\tilde N_k^2} \right)$ into \eqref{C13-b}, we have
\begin{align}\label{gamma_gamma}
\Upsilon  = \frac{{{P_{\max }}}}{{\sum\nolimits_{k = 1}^K {\left( {1/\left( {{\Gamma _k}\tilde N_k^2} \right)} \right)} }}.
\end{align}
Hence, the optimal $\left\{ {{p_{kL}}} \right\}$ under an arbitrary $\left\{ {{{\tilde N}_k}} \right\}$ is given by
\begin{align}\label{optinal_power allocation}
{p_{kL}} = \frac{{{P_{\max }}/\left( {{\Gamma _k}\tilde N_k^2} \right)}}{{\sum\nolimits_{k = 1}^K {\left( {1/\left( {{\Gamma _k}\tilde N_k^2} \right)} \right)} }},\forall k.
\end{align}
By further substituting \eqref{optinal_power allocation} into problem \eqref{C14}, problem \eqref{C14} is equivalently transformed into
\begin{align}\label{optimization_elements_maxmin}
\mathop {\min }\limits_{\left\{ {{{\tilde N}_k}} \right\}}~~ {\rm{ }}\sum\nolimits_{k = 1}^K {\frac{1}{{{\Gamma _k}\tilde N_k^2}}}~~{\rm{s}}{\rm{.t}}{\rm{.}}~~\eqref{C14-b}.
\end{align}
It can be readily verified that problem \eqref{optimization_elements_maxmin} is a convex optimization problem. Hence, its optimal solution can be derived by analyzing the KKT conditions. In particular, the Lagrangian function of problem \eqref{optimization_elements_maxmin} is given by
\begin{align}\label{Lagrangian1}
{{\cal L}_1}\left( {{{\tilde N}_k},\mu } \right) = \sum\nolimits_{k = 1}^K {\frac{1}{{{\Gamma _k}\tilde N_k^2}}}  + \mu \left( {\sum\nolimits_{k = 1}^K {{{\tilde N}_k}} - N } \right).
\end{align}
Its KKT conditions can be written as
\begin{align}\label{KKT_conditons1}
\!\!\frac{{\partial {{\cal L}_1}\left( {{{\tilde N}_k},\mu } \right)}}{{\partial {{\tilde N}_k}}} = \frac{{ - 2}}{{{\Gamma _k}\tilde N_k^3}} + \mu  = 0,\forall k,N - \sum\nolimits_{k = 1}^K {{{\tilde N}_k}}  = 0.
\end{align}
Based on \eqref{KKT_conditons1}, \eqref{optimal_solution_C14} can be obtained after some straightforward  manipulations.
\end{proof}

For the objective of maximizing the minimum user rate, Proposition 4 demonstrates that the number of IRS elements deployed in each cluster scales with ${\left( {\rho _{g,k}^{\rm{D}}\rho _{r,kL}^{\rm{D}}} \right)^{ - 2/3}}$. Note that ${\left( {\rho _{g,k}^{\rm{D}}\rho _{r,kL}^{\rm{D}}} \right)^{ - 2}}$ represents the concatenated path-loss of the BS-IRS-user $u_k^L$ link. The result is intuitive, since more elements have to be deployed in the cluster suffering the severe concatenated path-loss, which is helpful for balancing the SINR. Based on Proposition 1, the optimal solution for the original problem \eqref{C13} can be constructed via the integer rounding technique.

\subsubsection{Sum-rate Maximization}
We further consider the IRS element allocation design for the objective of the sum-rate maximization. The corresponding optimization problem of jointly optimizing the IRS element allocation and power allocation is formulated as follows:
\begin{subequations}\label{C15}
\begin{align}
\label{C15-a}\mathop {\max }\limits_{\left\{ {{p_k}} \right\},\left\{ {{N_k}} \right\}} \;\;&\sum\nolimits_{k = 1}^K {R_{kL}^{\rm{D}}}\\
\label{C15-b}{\rm{s.t.}}\;\;\;\;\;&\eqref{C13-b},\eqref{C13-c},\eqref{C13-d}.
\end{align}
\end{subequations}
Problem \eqref{C15} is challenging to solve due to both the coupled optimization variables in the objective function and to the discrete variables in constraint \eqref{C13-d}. To make problem \eqref{C15} tractable, we first consider to relax the discrete constraints on $\left\{ {{N_k}} \right\}$ in \eqref{C13-d}. Then, the resultant optimization problem is
\begin{subequations}\label{C16}
\begin{align}
\label{C16-a}\mathop {\max }\limits_{\left\{ {{p_k}} \right\},\left\{ {{N_k}} \right\}} \;\;&\sum\nolimits_{k = 1}^K {R_{kL}^{\rm{D}}}\\
\label{C16-b}{\rm{s.t.}}\;\;\;\;\;&{N_k} \ge 0,\\
\label{C16-b}&\eqref{C13-b},\eqref{C13-c}.
\end{align}
\end{subequations}
Then, we derive the asymptotically optimal solution of problem \eqref{C16} in the large-$N$ regime, which is formulated in the following proposition.
\begin{pos}
As $N \to \infty$, the asymptotically optimal solution of problem \eqref{C16}, denoted by $\left\{ {p_{kL}^*,N_k^*} \right\}$, is derived as
\begin{align}\label{optimal_solution2}
N_k^* = \frac{N}{K},p_{kL}^* = \frac{{{P_{\max }}}}{K}, \forall k \in {\cal K}.
\end{align}
\end{pos}
\begin{proof}
 Let ${N_k} = {\beta _k}N$. Under any given $\left\{ {{N_k}} \right\}$, the optimal $\left\{ {{p_{kL}}} \right\}$ can be derived as \cite{4657320}
\begin{align}\label{water_filling1}
{{\tilde p}_{kL}} = \max \left( {u - \frac{1}{{{\Gamma _k}\beta _k^2{N^2}}},0} \right)
\end{align}
with $\sum\limits_{k = 1}^K {{{\tilde p}_{kL}}}  = {P_{\max }}$, where ${{\Gamma _k}}$ is defined in \eqref{tao_k}. Let ${N_k} = {\beta _k}N$. Then, upon substituting \eqref{water_filling1} into problem \eqref{C16}, \eqref{C16} becomes equivalent to
\begin{subequations}\label{C17}
\begin{align}
\label{C17-a}\mathop {\max }\limits_{\left\{ {{\beta _k}} \right\}} \;\;&\sum\nolimits_{k = 1}^K {{{\log }_2}\left( {1 + {{\tilde p}_{kL}}\beta _k^2{N^2}{\Gamma _k}} \right)}\\
\label{C17-b}{\rm{s.t.}}\;\;&\sum\nolimits_{k = 1}^K {{\beta _k}}  = 1,\\
\label{C17-c}&{\beta _k} \ge 0.
\end{align}
\end{subequations}
For problem \eqref{C16}, we first focus on the case of ${\beta _k} > 0,\forall k$. As $N \to \infty $, ${{\tilde p}_{kL}} = u - 1/\left( {{\Gamma _k}\beta _k^2{N^2}} \right)$ holds naturally. Then, the objective function \eqref{C17-a} can be simplified as
\begin{align}\label{function17_large_N}
\mathop {\lim }\limits_{N \to \infty } \sum\nolimits_{k = 1}^K {{{\log }_2}\left( {1 + {{\tilde p}_{kL}}\beta _k^2{N^2}{\Gamma _k}} \right)} = \sum\nolimits_{k = 1}^K {{{\log }_2}\left( {u\beta _k^2{N^2}{\Gamma _k}} \right)}.
\end{align}
Hence, problem \eqref{C17} is reduced to
\begin{align}\label{case1_problem}
\mathop {\max }\limits_{\left\{ {{\beta _k}} \right\}}~~ \sum\nolimits_{k = 1}^K {{{\log }_2}\left( {u\beta _k^2{N^2}{\Gamma _k}} \right)} ~~{\rm{s}}{\rm{.t}}{\rm{.}}~~\eqref{C17-b}, \eqref{C17-c}.
\end{align}
Problem \eqref{case1_problem} is convex and its optimal solution can be shown to be ${\beta _k} = 1/K$. Accordingly, the optimal power allocation in this case is ${{\tilde p}_{kL}} = {P_{\max }}/K$ and its resultant objective value is
\begin{align}\label{objective_value1}
\tilde R = \sum\nolimits_{k = 1}^K {{{\log }_2}\left( {\frac{{{P_{\max }}{N^2}}}{{{K^3}}}{\Gamma _k}} \right)}.
\end{align}
Let $\bar {\cal K} \buildrel \Delta \over = \left\{ {k:{\beta _k} = 0,k \in {\cal K}} \right\}$. Then, we can show that $\bar {\cal K} \ne \emptyset $ is always suboptimal. Under any given ${\bar {\cal K}}$, the objective value can be similarly derived as
\begin{align}\label{objective_value1}
\mathord{\buildrel{\lower3pt\hbox{$\scriptscriptstyle\smile$}}
\over R}  = \sum\nolimits_{k \in {\cal K}\backslash \bar {\cal K}} {{{\log }_2}\left( {\frac{{{P_{\max }}{N^2}}}{{{{\left( {K - \left| {\bar {\cal K}} \right|} \right)}^3}}}{\Gamma _k}} \right)}.
\end{align}
It can be shown that
\begin{align}\label{case1_versus_case2}
\mathop {\lim }\limits_{N \to \infty } \frac{{\tilde R}}{{\mathord{\buildrel{\lower3pt\hbox{$\scriptscriptstyle\smile$}}
\over R} }} = \frac{K}{{K - \bar {\cal K}}} \ge 1,
\end{align}
which implies that the optimal solution of problem \eqref{C17} is always under the case of ${\beta _k} > 0,\forall k$ as $N \to \infty $. Thus, we complete the proof.
\end{proof}

Proposition 5 unveils that equal elements allocation is able to achieve the near-optimal performance under the large number of IRS elements regime. In a general multi-stream transmission, it is well known that the equal power allocation is asymptotically optimal in the high SNR region \cite{4657320}. Deploying a large number of IRS elements generates the equivalent high SNR regime artificially, which makes the equal power/element allocation scheme is asymptotically optimal for a high $N$.

\subsection{Performance Characterization under Rician Channels}
In this subsection, we characterize the ergodic rate of the proposed hybrid SDMA-TDMA scheme under the Rician fading channels. Under the assumption of $\delta _1^{\rm{D}}{\rm{ = }} \ldots {\rm{ = }}\delta _K^{\rm{D}} = {\delta ^{\rm{D}}}$ and $\varepsilon _k^{\rm{D}} \to \infty, \forall k$, the closed-form expression of the ergoidic rate of user $u_k^l$ is approximately derived in the following proposition.

\begin{pos}
In the $l$-th TS, the ergodic rate of user $u_k^l$ can be approximated as
\begin{align}\label{ergodic_rate_result}
{\rm{E}}\left[ {{R_{kl}}} \right] \!\!\approx \!\!\frac{{{p_{kl}}{{\left( {\rho _{g,k}^{\rm{D}}\rho _{r,kl}^{\rm{D}}} \right)}^2}\left( {\eta MN_k^2{{\left( {\frac{{{2^b}}}{\pi }\sin \frac{\pi }{{{2^b}}}} \right)}^2}{\rm{ + }}\left( {1 - \eta } \right){N_k}} \right)}}{{{{\left( {\rho _{g,k}^{\rm{D}}\rho _{r,kl}^{\rm{D}}} \right)}^2}\sum\nolimits_{i = 1,i \ne k}^K {{p_{il}}\left( {1 - \eta } \right){N_k} + {\sigma ^2}} }},
\end{align}
where $\eta  = {\delta ^{\rm{D}}}/\left( {{\delta ^{\rm{D}}} + 1} \right)$.
\end{pos}
\begin{proof}
By applying Lemma 1 in \cite{6816003}, the ergodic rate of user $u_k^l$ can be approximated as
\begin{align}\label{ergodic_rate_pre}
{\rm{E}}\left[ {{R_{kl}}} \right] \!\!\approx\!\! {\log _2}\left( {1 + \frac{{{\mathop{\rm E}\nolimits} \left[ {{{\left| {{{\left( {{\bf{h}}_{kl}^{\rm{D}}\left( {{\bf{\Theta }}_k^{\rm{D}}\left[ l \right]} \right)} \right)}^H}{{\bf{w}}_k}\left[ l \right]} \right|}^2}} \right]}}{{\sum\limits_{i = 1,i \ne k}^K {{\mathop{\rm E}\nolimits} {{\left| {{{\left( {{\bf{h}}_{kl}^{\rm{D}}\left( {{\bf{\Theta }}_k^{\rm{D}}\left[ l \right]} \right)} \right)}^H}{{\bf{w}}_i}\left[ l \right]} \right|}^2} + {\sigma ^2}} }}} \right).
\end{align}
Upon substituting $\left\{ {{{\bf{w}}_k}\left[ l \right],{\bf{\Theta }}_k^{\rm{D}}\left[ l \right]} \right\}$ given in \eqref{joint_bf_configuration} into \eqref{ergodic_rate_pre}, the signal term ${{\mathop{\rm E}\nolimits} \left[ {{{\left| {{{\left( {{\bf{h}}_{kl}^{\rm{D}}\left( {{\bf{\Theta }}_k^{\rm{D}}\left[ l \right]} \right)} \right)}^H}{{\bf{w}}_k}\left[ l \right]} \right|}^2}} \right]}$ can be derived as
\begin{align}\label{average_intended_signal}
\begin{array}{l}
{\rm{E}}\left[ {{{\left| {{{\left( {{\bf{h}}_{kl}^{\rm{D}}\left( {{\bf{\Theta }}_k^{\rm{D}}\left[ l \right]} \right)} \right)}^H}{{\bf{w}}_k}\left[ l \right]} \right|}^2}} \right]\\
 = {p_{kl}}{\left( {\rho _{g,k}^{\rm{D}}\rho _{r,kl}^{\rm{D}}} \right)^2}{\rm{E}}\left[ {{{\left| {{{\left( {{\bf{\bar h}}_{r,kl}^{\rm{D}}} \right)}^H}{\bf{\Theta }}_k^{\rm{D}}\left[ l \right]\left( {\sqrt \eta  {\bf{\bar G}}_k^{\rm{D}} + \sqrt {1 - \eta } {\bf{\tilde G}}_k^{\rm{D}}} \right){{\bf{w}}_k}\left[ l \right]} \right|}^2}} \right]\\
 = {p_{kl}}{\left( {\rho _{g,k}^{\rm{D}}\rho _{r,kl}^{\rm{D}}} \right)^2}\left( \begin{array}{l}
\eta {\rm{E}}\left[ {{{\left| {{{\left( {{\bf{\bar h}}_{r,kl}^{\rm{D}}} \right)}^H}{\bf{\Theta }}_k^{\rm{D}}\left[ l \right]{\bf{\bar G}}_k^{\rm{D}}{{\bf{w}}_k}\left[ l \right]} \right|}^2}} \right]\\
 + \left( {1 - \eta } \right){\rm{E}}\left[ {{{\left| {{{\left( {{\bf{\bar h}}_{r,kl}^{\rm{D}}} \right)}^H}{\bf{\Theta }}_k^{\rm{D}}\left[ l \right]{\bf{\tilde G}}_k^{\rm{D}}{{\bf{w}}_k}\left[ l \right]} \right|}^2}} \right]
\end{array} \right)\\
 = {p_{kl}}{\left( {\rho _{g,k}^{\rm{D}}\rho _{r,kl}^{\rm{D}}} \right)^2}\left( {\eta MN_k^2{{\left( {\frac{{{2^b}}}{\pi }\sin \frac{\pi }{{{2^b}}}} \right)}^2}{\rm{ + }}\left( {1 - \eta } \right){N_k}} \right).
\end{array}
\end{align}
Then, the inter-user interference term is obtained as
\begin{align}\label{average_interference}
\begin{array}{*{20}{l}}
{{\rm{E}}{{\left| {{{\left( {{\bf{h}}_{kl}^{\rm{D}}\left( {{\bf{\Theta }}_k^{\rm{D}}\left[ l \right]} \right)} \right)}^H}{{\bf{w}}_i}\left[ l \right]} \right|}^2}}\\
{ = {p_{il}}{{\left( {\rho _{g,k}^{\rm{D}}\rho _{r,kl}^{\rm{D}}} \right)}^2}{\rm{E}}\left[ {{{\left| {{{\left( {{\bf{\bar h}}_{r,kl}^{\rm{D}}} \right)}^H}{\bf{\Theta }}_k^{\rm{D}}\left[ l \right]\left(
\sqrt \eta  {\bf{\bar G}}_k^{\rm{D}}
 + \sqrt {1 - \eta } {\bf{\tilde G}}_k^{\rm{D}}
 \right){{\bf{w}}_i}\left[ l \right]} \right|}^2}} \right]}\\
{ = {p_{il}}{{\left( {\rho _{g,k}^{\rm{D}}\rho _{r,kl}^{\rm{D}}} \right)}^2}\left( {1 - \eta } \right){\rm{E}}\left[ {{{\left| {{{\left( {{\bf{\bar h}}_{r,kl}^{\rm{D}}} \right)}^H}{\bf{\Theta }}_k^{\rm{D}}\left[ l \right]{\bf{\tilde G}}_k^{\rm{D}}{{\bf{w}}_i}\left[ l \right]} \right|}^2}} \right]}\\
{ = {p_{il}}{{\left( {\rho _{g,k}^{\rm{D}}\rho _{r,kl}^{\rm{D}}} \right)}^2}\left( {1 - \eta } \right){N_k}.}
\end{array}
\end{align}
substituting \eqref{average_intended_signal} and \eqref{average_interference} into \eqref{ergodic_rate_pre}, the proof is completed.
\end{proof}

The accuracy of \eqref{ergodic_rate_result} will be verified by Monte Carlo simulations in the next subsection. Proposition 6 provides an efficient way of quantifing the performance loss of the achievable rate under the Rician fading channel relative to the LoS channel. It is observed from \eqref{ergodic_rate_result} that the residual inter-user interference increases linearly with ${{N_k}}$ due to the NLoS component in the BS-IRS link.
\vspace{-10pt}
\subsection{Numerical Results}
In this subsection, we first examine the effectiveness of the proposed IRS element allocation strategies numerically. Then, we further quantify the performance of the \emph{distributed IRS} under the Rician fading channel. Unless otherwise stated, we set $K = 2$ and other parameters are same as those for Fig. 2-Fig. 5.

\subsubsection{Minimum User Rate Maximization}
We consider a heterogeneous channel setup, where the concatenated path-loss for the $K$ users is set as $\left[ { - 140, - 150} \right]$ dB. The following schemes are considered for comparison: 1) \textbf{Distributed IRS-o}: The optimal design for IRS elements and power allocation provided in Proposition 4; 2) \textbf{Distributed IRS-eq}: The optimal power allocation is performed under the identical IRS elements allocation of ${N_k} = N/K,\forall k$; 3) \textbf{Centralized IRS}: Optimal power allocation under the \emph{centralized IRS} architecture.
%

\begin{figure}[t]
\begin{minipage}[t]{0.4\linewidth}
\centering
\includegraphics[width=2.9in]{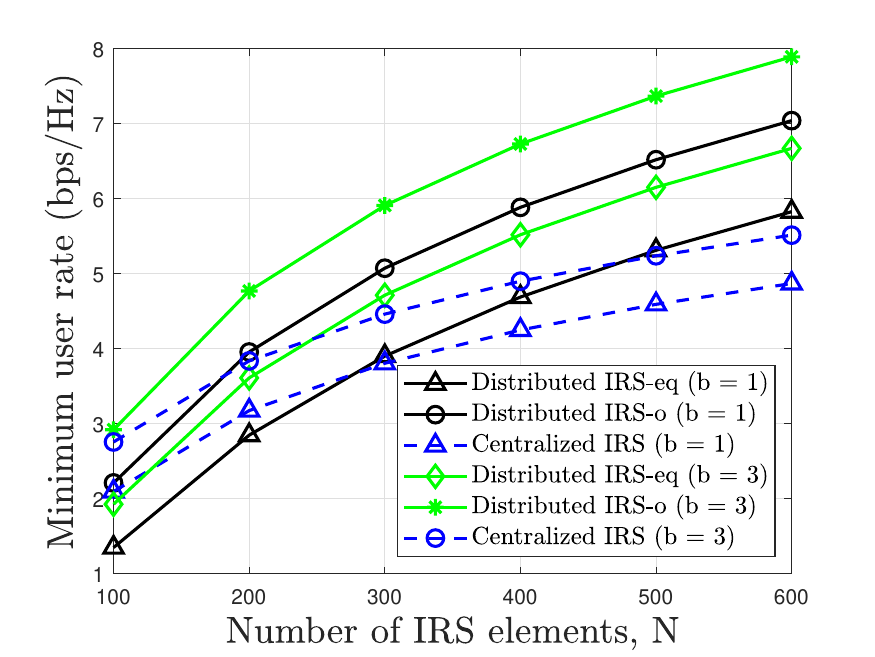}
\caption{Minimum user rate versus $N$ under the heterogeneous channel setup.}
\label{elements_allocation1}
\end{minipage}%
\hfill
\begin{minipage}[t]{0.4\linewidth}
\centering
\includegraphics[width=2.9in]{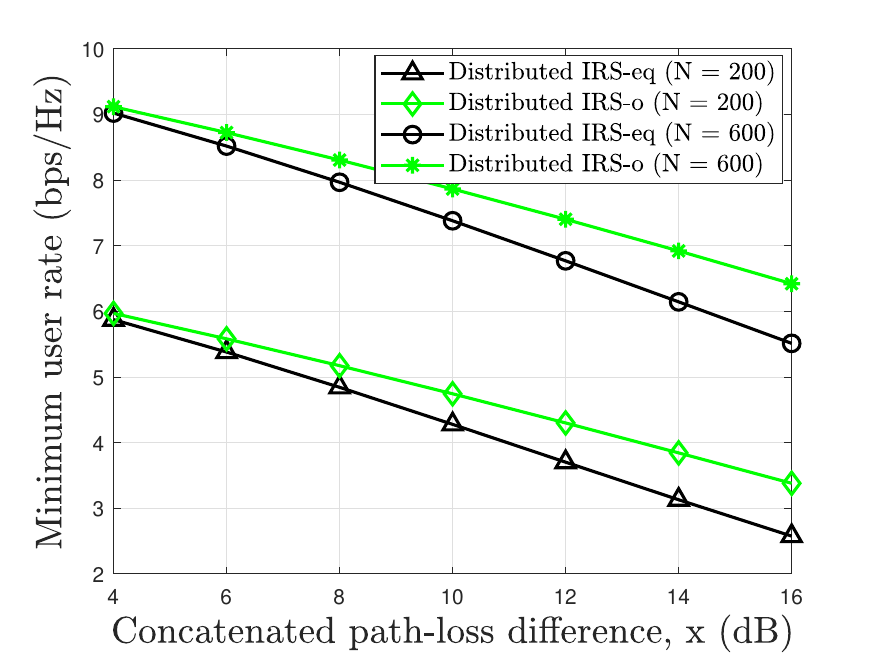}
\caption{Minimum user rate versus concatenated path-loss difference.}
\label{elements_allocation2}
\end{minipage}
\vspace{-16pt}
\end{figure}

In Fig. \ref{elements_allocation1}, we plot the minimum user rate versus $N$ under our heterogeneous channel setup. First, it is observed that the proposed IRS elements allocation design significantly improves the minimum user rate as compared to the case of the identical IRS element allocation of ${N_k} = N/K$. This is expected since the effective channel power gains of different user clusters tend to be homogenous due to flexibly allocating the IRS elements. This suggests that the rate fairness issue can be alleviated by appropriate IRS elements allocation design. Note that careful IRS element allocation is capable of mitigating the severe concatenated path-loss of the user clusters, which are located far from the BS. Moreover, the $N$ required for distributed IRS to outperform centralized IRS can be reduced via the optimal IRS element allocation design, which effectively enlarges the operating region of the distributed IRS.

In Fig. \ref{elements_allocation2}, we study the impact of the concatenated path-loss difference on the minimum user rate, by plotting it versus the concatenated path-loss difference, denoted by $x$. For the given $x$, the corresponding minimum user rate is obtained under the concatenated path-loss $\left[ { - 140, - 140 - x} \right]$ dB. It is observed that the minimum user rate under the optimal element allocation moderately decreases with $x$, while that of the identical element allocation decreases sharply with $x$. This highlights the importance of carefully optimizing the IRS element allocation for high path-loss differences among users.

\begin{figure}[t]
\begin{minipage}[t]{0.45\linewidth}
\centering
\includegraphics[width=2.9in]{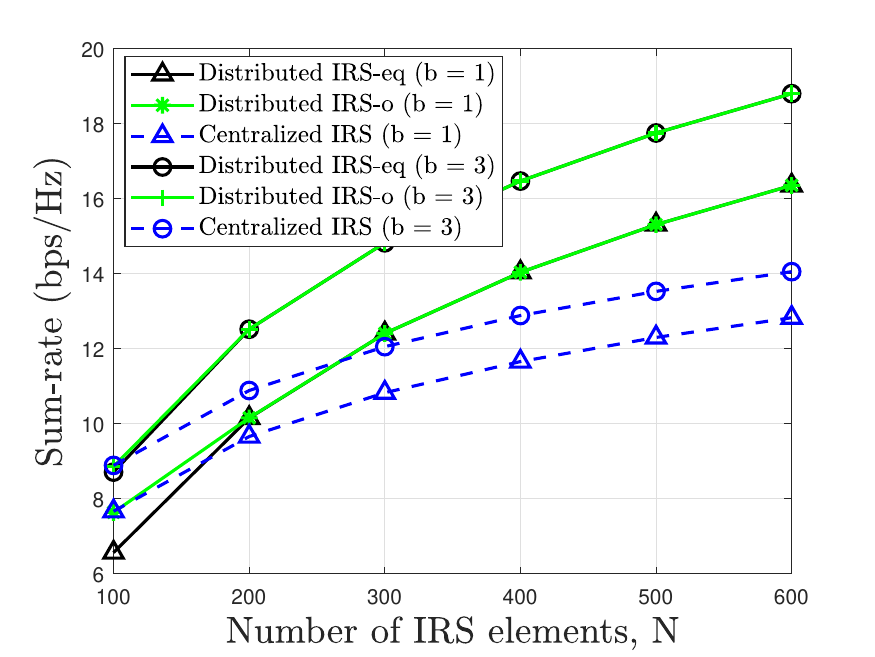}
\caption{Sumrate versus $N$ under the heterogeneous channel setup.}
\label{elements_allocation3}
\end{minipage}%
\hfill
\begin{minipage}[t]{0.45\linewidth}
\centering
\includegraphics[width=2.9in]{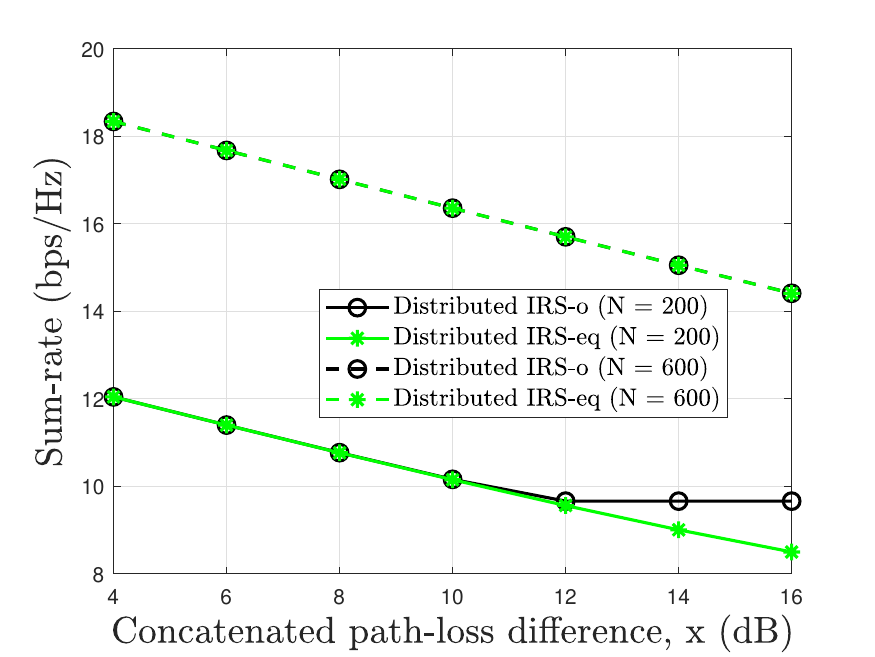}
\caption{Sumrate versus concatenated path-loss difference.}
\label{elements_allocation4}
\end{minipage}
\vspace{-16pt}
\end{figure}

\subsubsection{Objective of Sum-Rate Maximization}
%
We adopt the same parameters as those in Fig. \ref{elements_allocation1} and Fig. \ref{elements_allocation2}. The following schemes are considered: 1) \textbf{Distributed IRS-eq}: Both the equal power allocation and identical IRS element allocation are adopted under the \emph{distributed IRS}; 2) \textbf{Distributed IRS-o}: Exhaustive search is employed to find the optimal element allocation; 3) \textbf{Centralized IRS}: The maximum sum-rate of the \emph{centralized IRS} is achieved by only scheduling the user having the maximum received SNR, i.e, $R_c^s = \mathop {\max }\limits_{k \in \left\{ {1, \ldots ,K} \right\}} {\log _2}\left( {1 + \gamma _k^c} \right)$ with ${\gamma _k^c}$ representing the received SNR at user $k$. In Fig. \eqref{elements_allocation3}, we show the sum-rate versus the total number of IRS elements. As $N>200$, it is observed from Fig. \eqref{elements_allocation3} that the sum-rate achieved by equal power and identical IRS element allocations is almost consistent with that achieved by exhaustive search, which validates our analysis in Proposition 5.

In Fig. \ref{elements_allocation4}, we examine the ergodic sum-rate versus the concatenated path-loss difference, i.e., $x$. For the case of $N = 200$, we observe that equal power and identical IRS element allocation is able to achieve near-optimal performance for $x \le 12$ dB. However, when $x > 12$ dB, the sum-rate achieved by exhaustive search remains fixed as $x$ increases. This is due to the fact that all the elements are placed near user 1 for maximizing the sum-rate, since the concatenated path-loss of user 2 is significantly lower than that of user 1. In this case, the sum-rate reduces to the achieved rate of user 1, since user 2 is not scheduled, which leads to a severe user fairness issue. Nevertheless, for the case of $N = 600$, we can observe that equal power and identical IRS element allocation achieves almost the same performance as that of exhaustive search even at $x=16$ dB. The result unveils that the user fairness issue caused by high concatenated path-loss differences can be addressed by deploying more IRS elements.

\subsubsection{Performance Evaluation Under Rician Channels}

\begin{figure}[!t]
\setlength{\abovecaptionskip}{-5pt}
\setlength{\belowcaptionskip}{-5pt}
\centering
\includegraphics[width= 0.4\textwidth]{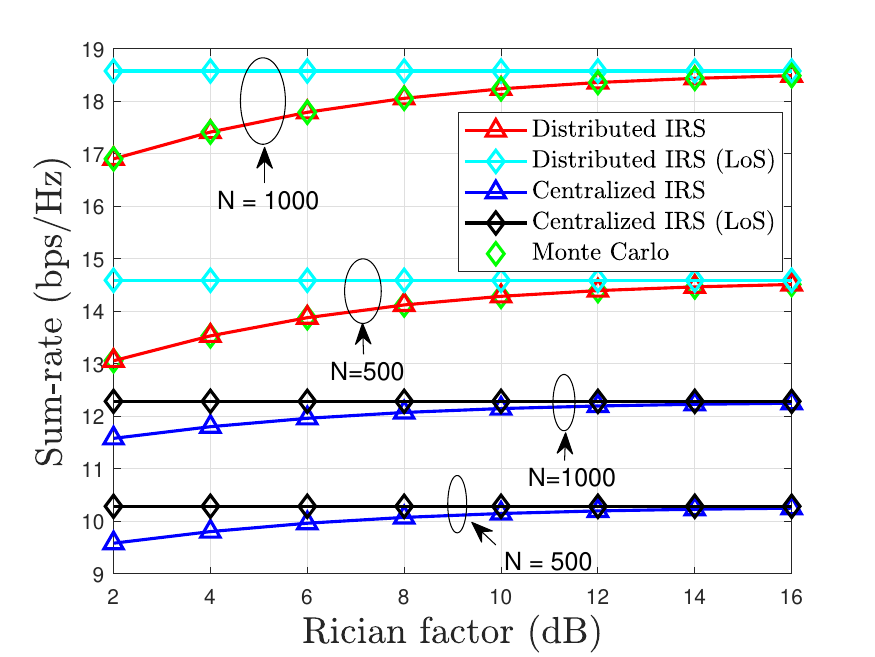}
\DeclareGraphicsExtensions.
\caption{Sum-rate versus Rician factor.}
\label{Rician_fading}
\vspace{-16pt}
\end{figure}
In Fig. \ref{Rician_fading}, we plot the sum-rate versus the Rician factor. For examining the accuracy of the analysis in Proposition 6, Monte-Carlo simulations are implemented and the results are obtained by averaging those of 10000 channel realizations. Regarding the case of \emph{centralized IRS}, the BS-IRS link is assumed to be the pure LoS channel and the Rician factor of the IRS-user links are set to be the same as that of the BS-IRSs links under the \emph{distributed IRS}. The sum-rates under the \emph{centralized IRS} are obtained by simulations. As shown in Fig. \ref{Rician_fading}, the results of the expression in Proposition 6 are tightly matched with the results obtained by the Monte-Carlo simulations, which validates the accuracy of the approximation. It is observed from Fig. \ref{Rician_fading} that the sum-rate under the \emph{distributed IRS} is more sensitive to the Rician factor than that of the \emph{centralized IRS}. Note that the NLoS component of channels degrades the passive beamforming gain and also increases the inter-user interference under the \emph{distributed IRS}. By contrast, the NLoS component of channels only reduces the passive beamforming gain for the \emph{centralized IRS}. The result highlights the importance of deploying distributed IRSs to create strong LoS links with the BS. Nevertheless, the sum-rates of the \emph{distributed IRS} are still higher than those of the \emph{centralized IRS} for a wide range of Rician factors. This demonstrates the superiority of the \emph{distributed IRS} architecture in terms of network capacity.

\vspace{-10pt}
\section{Conclusions}
In this paper, we investigated the capacity of the multi-antenna broadcast channel assisted by both the \emph{distributed IRS} and \emph{centralized IRS} deployment architectures. We provided an analytical framework to theoretically compare the capacity achieved by the \emph{distributed} and \emph{centralized IRS}. By capturing the fundamental tradeoff between the spatial multiplexing gain and passive beamforming gain, we analytically demonstrated that the \emph{distributed IRS} is capable of outperforming \emph{centralized IRS} when the total number of IRS elements is higher than a threshold. Furthermore, to fully unleash the potential of spatial multiplexing and dynamic IRS beamforming, we proposed an efficient hybrid SDMA-TDMA scheme for the \emph{distributed IRS}. Moreover, we studied the IRS element allocation problem under the \emph{distributed IRS} to customize channels for both minimum user rate maximization and sum-rate maximization. Our numerical results validated the theoretical findings and demonstrated the benefits of the \emph{distributed IRS} for improving the system capacity under various setups.
\vspace{-10pt}
\section*{Appendix A: \textsc{Proof of Proposition 1}}
To obtain \eqref{C_D_closed}, we first derive the outer bound of ${{\cal C}^{\rm{D}}}$ and then we show that this upper is tight and can be achieved. By removing the inter-user interference, it can be shown that $R_k^{\rm{D}}\left( {\left\{ {{{\bf{w}}_k}} \right\},{\bf{\Theta }}_k^{\rm{D}}} \right)$ is upper-bounded by
\begin{align}\label{upper_bound}
R_k^{\rm{D}}\left( {\left\{ {{{\bf{w}}_k}} \right\},{\bf{\Theta }}_k^{\rm{D}}} \right) \le \hat R_k^{\rm{D}}\left( {\left\{ {{{\bf{w}}_k}} \right\},{\bf{\Theta }}_k^{\rm{D}}} \right){\rm{ = }}{\log _2}\left( {1 + {{{{\left| {{{\left( {{\bf{h}}_k^{\rm{D}}\left( {{\bf{\Theta }}_k^{\rm{D}}} \right)} \right)}^H}{{\bf{w}}_k}} \right|}^2}} \mathord{\left/
 {\vphantom {{{{\left| {{{\left( {{\bf{h}}_k^{\rm{D}}\left( {{\bf{\Theta }}_k^{\rm{D}}} \right)} \right)}^H}{{\bf{w}}_k}} \right|}^2}} {{\sigma ^2}}}} \right.
 \kern-\nulldelimiterspace} {{\sigma ^2}}}} \right).
\end{align}
Let ${{\bf{w}}_k}{\rm{ = }}\sqrt {{p_k}} {{{\bf{\tilde w}}}_k}$ with $\left\| {{{{\bf{\tilde w}}}_k}} \right\|_2^2 = 1$. Then, we derive the outer bound of ${{\cal C}^{\rm{D}}}$ by analytically solving the following optimization problem:
\begin{align}\label{max_power_gain}
\mathop {\max }\limits_{{{{\bf{\tilde w}}}_k},{\bf{\Theta }}_k^{\rm{D}}}~ {p_k}{\left| {{{\left( {{\bf{h}}_k^{\rm{D}}\left( {{\bf{\Theta }}_k^{\rm{D}}} \right)} \right)}^H}{{{\bf{\tilde w}}}_k}} \right|^2}~{\rm{s}}{\rm{.t}}{\rm{.}}~\left\| {{{{\bf{\tilde w}}}_k}} \right\|_2^2 = 1,{\bf{\Theta }}_k^{\rm{D}} \in {{\cal P}^{\rm{D}}}.
\end{align}
By exploiting the special structure of ${\left( {{\bf{h}}_k^{\rm{D}}\left( {{\bf{\Theta }}_k^{\rm{D}}} \right)} \right)^H} = \rho _{g,k}^{\rm{D}}\rho _{r,k}^{\rm{D}}{\left( {{\bf{\bar h}}_{r,k}^{\rm{D}}} \right)^H}{\bf{\Theta }}_k^{\rm{D}}{\bf{\bar G}}_k^{\rm{D}}$ and ${\bf{\bar G}}_k^{\rm{D}}$ in \eqref{channel_Gk_LoS}, problem \eqref{max_power_gain} can be equivalently decomposed into two parallel sub-problems as follow:
\begin{align}\label{w_k_direction}
\mathop {\max }\limits_{{{{\bf{\tilde w}}}_k}}~~ {\left| {{\bf{a}}_M^H\left( {\sin \theta _{{\rm{T}},k}^{{\rm{AOD}}}} \right){{{\bf{\tilde w}}}_k}} \right|^2}~~{\rm{s}}{\rm{.t}}{\rm{.}}~~\left\| {{{{\bf{\tilde w}}}_k}} \right\|_2^2 = 1,
\end{align}
\begin{align}\label{IRS_pattern_distributed}
\mathop {\max }\limits_{{\bf{\Theta }}_k^{\rm{D}}} {\left| {{{\left( {{\bf{\bar h}}_{r,k}^{\rm{D}}} \right)}^H}{\bf{\Theta }}_k^{\rm{D}}{{\bf{a}}_{{\rm{S}},k}}} \right|^2}~~{\rm{s}}{\rm{.t}}{\rm{.}}~~{\bf{\Theta }}_k^{\rm{D}} \in {{\cal P}^{\rm{D}}}.
\end{align}
For problem \eqref{w_k_direction}, the optimal ${{{{\bf{\tilde w}}}_k}}$ is ${\bf{\tilde w}}_k^* = {{\bf{a}}_M}\left( {\sin \theta _{{\rm{T}},k}^{{\rm{AOD}}}} \right)/\sqrt M $ and its associated optimal objective value is $M$. For problem \eqref{IRS_pattern_distributed}, the optimal ${\bf{\Theta }}_k^{\rm{D}}$ is given by $\theta _{k,n}^{{\rm{D*}}} = \mathop {\arg \min }\nolimits_{\theta _{k,n}^D \in {\cal F}} \left| {{e^{j\theta _{k,n}^{\rm{D}}}} - {e^{j\theta _{k,n}^{{\rm{D}}**}}}} \right|$, where $\theta _{k,n}^{{\rm{D}}**} = \arg \left( {{{\left[ {{\bf{\bar h}}_{r,k}^{\rm{D}}} \right]}_n}} \right) - \arg \left( {{{\left[ {{{\bf{a}}_{{\rm{S}},k}}} \right]}_n}} \right)$. As $N \to \infty$, the optimal objective value of problem \eqref{IRS_pattern_distributed} is derived as
\begin{align}\label{IRS_BF_gain}
{\left| {{{\left( {{\bf{\bar h}}_{r,k}^{\rm{D}}} \right)}^H}{\bf{\Theta }}_k^{\rm{D}}{{\bf{a}}_{{\rm{S}},k}}} \right|^2} &= {\left| {\sum\nolimits_{n = 1}^N {{e^{j\left( {\theta _{k,n}^{{\rm{D*}}} - \theta _{k,n}^{{\rm{D}}**}} \right)}}} } \right|^2}
 = {N^2}{\left| {\frac{1}{N}\sum\nolimits_{n = 1}^N {{e^{j\left( {\theta _{k,n}^{{\rm{D*}}} - \theta _{k,n}^{{\rm{D}}**}} \right)}}} } \right|^2}\nonumber\\
&\mathop  \to \limits^{a.s.} {N^2}{\rm{E}}\left[ {{e^{j\left( {\theta _{k,n}^{{\rm{D*}}} - \theta _{k,n}^{{\rm{D}}**}} \right)}}} \right]\mathop  = \limits^{\left( a \right)} {N^2}{\left( {\frac{{{2^b}}}{\pi }\sin \frac{\pi }{{{2^b}}}} \right)^2},
\end{align}
where (a) is valid because $\left( {\theta _{k,n}^{{\rm{D*}}} - \theta _{k,n}^{{\rm{D}}**}} \right)$ is uniformly distributed in $\left[ { - {\pi  \mathord{\left/
 {\vphantom {\pi  {{2^b},{\pi  \mathord{\left/
 {\vphantom {\pi  {{2^b}}}} \right.
 \kern-\nulldelimiterspace} {{2^b}}}}}} \right.
 \kern-\nulldelimiterspace} {{2^b},{\pi  \mathord{\left/
 {\vphantom {\pi  {{2^b}}}} \right.
 \kern-\nulldelimiterspace} {{2^b}}}}}} \right]$. Hence, the objective value of problem \eqref{optimal_beamforming} is ${p_k}M{N^2}{\left( {\frac{{{2^b}}}{\pi }\sin \frac{\pi }{{{2^b}}}} \right)^2}{\left( {\rho _{g,k}^{\rm{D}}} \right)^2}{\left( {\rho _{r,k}^{\rm{D}}} \right)^2}$. Correspondingly, the outer bound of ${{\cal C}^{\rm{D}}}$, denoted by ${\cal C}_o^{\rm{D}}$ is given by
 \begin{align}\label{C_D_O}
{\cal C}_o^{\rm{D}} \!\! =\!\! \left\{ {{{\bf{r}}^{\rm{D}}}:0 \le r_k^{\rm{D}} \le {{\log }_2}\left( {1 + \frac{{{p_k}M{N^2}{{\left( {\frac{{{2^b}}}{\pi }\sin \frac{\pi }{{{2^b}}}} \right)}^2}}}{{{K^2}{\sigma ^2}{{\left( {\rho _{g,k}^{\rm{D}}\rho _{r,k}^{\rm{D}}} \right)}^{ - 2}}}}} \right)} \right\},
\end{align}
with $\sum\nolimits_{k = 1}^K {{p_k} = {P_{\max }}}$.

Then, we show that the outer bound ${\cal C}_o^{\rm{D}}$ is tight. Under the condition that \eqref{deployment_condition} is satisfied, it can be readily verified that
\begin{align}\label{inner_product}
{\left| {{\bf{a}}_M^H\left( {\sin \theta _{{\rm{T}},i}^{{\rm{AOD}}}} \right){{\bf{a}}_M}\left( {\sin \theta _{{\rm{T}},k}^{{\rm{AOD}}}} \right)} \right|^2} &= \sum\nolimits_{m = 1}^M {{e^{j\frac{{2\pi d}}{\lambda }\left( {\sin \theta _{{\rm{T}},k}^{{\rm{AOD}}} - \sin \theta _{{\rm{T}},i}^{{\rm{AOD}}}} \right)\left( {M - 1} \right)}}} \nonumber\\
& = {\left| {\frac{{\sin \frac{{\pi d}}{\lambda }M\left( {\sin \theta _{{\rm{T}},k}^{{\rm{AOD}}} - \sin \theta _{{\rm{T}},i}^{{\rm{AOD}}}} \right)}}{{\sin \frac{{\pi d}}{\lambda }\left( {\sin \theta _{{\rm{T}},k}^{{\rm{AOD}}} - \sin \theta _{{\rm{T}},i}^{{\rm{AOD}}}} \right)}}} \right|^2} = 0,\forall k \ne i.
\end{align}
Based on \eqref{inner_product} and setting ${\bf{\tilde w}}_k^* = {{\bf{a}}_M}\left( {\sin \theta _{{\rm{T}},k}^{{\rm{AOD}}}} \right)/\sqrt M$, we have ${\left| {{{\left( {{\bf{h}}_k^{\rm{D}}\left( {{\bf{\Theta }}_k^{\rm{D}}} \right)} \right)}^H}{{\bf{w}}_i}} \right|^2} = 0,\forall k \ne i$, which indicates that the inter-user interference can be perfectly nulled. Hence, we have ${\cal C}^{\rm{D}}={\cal C}_o^{\rm{D}}$ and thus \eqref{C_D_closed} is obtained. For maximizing the system's sum rate, we formulate the following optimization problem:
\begin{align}\label{sum_rate_maximization21}
&\mathop {\max }\nolimits_{\left\{ {{p_k}} \right\}} \sum\limits_{k = 1}^K {{{\log }_2}\left( {1 + \frac{{{p_k}M{N^2}{{\left( {\rho _{g,k}^{\rm{D}}} \right)}^2}{{\left( {\rho _{r,k}^{\rm{D}}} \right)}^2}}}{{{K^2}{\sigma ^2}}}{{\left( {\frac{{{2^b}}}{\pi }\sin \frac{\pi }{{{2^b}}}} \right)}^2}} \right)}\nonumber\\
&~~{\rm{s}}{\rm{.t}}{\rm{.}}~~\sum\nolimits_{k = 1}^K {{p_k} = {P_{\max }}}.
\end{align}
The optimal ${\left\{ {{p_k}} \right\}}$ follows the well-known water-filling power allocation \cite{4657320}, which is given by
\begin{align}\label{sum_rate_maximization}
p_k^* = \max \left( {u - \frac{{{\sigma ^2}}}{{M{N^2}{{\left( {\rho _{g,k}^{\rm{D}}} \right)}^2}{{\left( {\rho _{r,k}^{\rm{D}}} \right)}^2}{{\left( {\frac{{{2^b}}}{\pi }\sin \frac{\pi }{{{2^b}}}} \right)}^2}}},0} \right).
\end{align}
Based on \eqref{assumption1} in \emph{Assumption 1}, i.e., ${\left( {\rho _{g,1}^{\rm{D}}} \right)^2}{\left( {\rho _{r,1}^{\rm{D}}} \right)^2} =  \ldots  = {\left( {\rho _{g,K}^{\rm{D}}} \right)^2}{\left( {\rho _{r,K}^{\rm{D}}} \right)^2}$, we further have $p_k^* = {P_{\max }}/K$, which leads to \eqref{sumrate_D_closed} and  \eqref{optimal_beamforming}. Thus, we complete the proof.

\section*{Appendix B: \textsc{Proof of Proposition 2}}
For the LoS channel scenario, we have
\begin{align}\label{channel_LoS_centralized}
{\left( {{\bf{h}}_k^{\rm{C}}\left( {{{\bf{\Theta }}^{\rm{C}}}} \right)} \right)^H} = \rho _g^{\rm{C}}\rho _{r,k}^{\rm{C}}{\left( {{\bf{\bar h}}_{r,kl}^{\rm{C}}} \right)^H}{{\bf{\Theta }}^{\rm{C}}}{{{\bf{\bar G}}}^{\rm{C}}}= \rho _g^{\rm{C}}\rho _{r,k}^{\rm{C}}\underbrace {\left( {{{\left( {{\bf{\bar h}}_{r,k}^{\rm{C}}} \right)}^H}{{\bf{\Theta }}^{\rm{C}}}{{\bf{a}}_{\rm{S}}}} \right)}_{{\rm{scalar}}}{\bf{a}}_M^H\left( {\sin \theta _{\rm{T}}^{{\rm{AOD}}}} \right),\forall k,
\end{align}
which indicates that all ${{\bf{h}}_k^{\rm{C}}\left( {{{\bf{\Theta }}^{\rm{C}}}} \right)}$'s are linearly dependent and parallel with ${{\bf{a}}_M}\left( {\sin \theta _{\rm{T}}^{{\rm{AOD}}}} \right)$. According to the uplink-downlink duality for this degraded broadcast channel, all achievable rate tuples satisfy the following condition:
\begin{align}\label{rate_constrained_centralized}
\sum\limits_{k \in {\cal J}} {r_k^{\rm{C}} \le } {\log _2}\left( {1 \!+\! \frac{{\sum\nolimits_{k \in {\cal J}} {{p_k}{{\left| {{{\left( {{\bf{h}}_k^{\rm{C}}\left( {{{\bf{\Theta }}^{\rm{C}}}} \right)} \right)}^H}{{{\bf{\tilde w}}}_k}} \right|}^2}} }}{{{\sigma ^2}}}} \right),\forall {\cal J} \subseteq {\cal K},
\end{align}
with $\sum\nolimits_{k \in {\cal J}} {{p_k}}  \le {P_{\max }}$ and $\left\| {{{{\bf{\tilde w}}}_k}} \right\|_2^2 = 1$.
Then, we derive the outer bound of ${{\cal C}^{\rm{C}}}$, denoted by ${\cal C}_{\rm{o}}^{\rm{C}}$, and further show that ${\cal C}_{\rm{o}}^{\rm{C}}$ is tight. The outer bound ${\cal C}_{\rm{o}}^{\rm{C}}$ is derived by considering the following set of optimization problem:
\begin{align}\label{max_power_gain_centralized}
\mathop {\max }\limits_{{{{\bf{\tilde w}}}_k},{\bf{\Theta }}^{\rm{C}}}~~ {\left| {{{\left( {{\bf{h}}_k^{\rm{D}}\left( {{\bf{\Theta }}^{\rm{C}}} \right)} \right)}^H}{{{\bf{\tilde w}}}_k}} \right|^2}~~{\rm{s}}{\rm{.t}}{\rm{.}}~~\left\| {{{{\bf{\tilde w}}}_k}} \right\|_2^2 = 1,{\bf{\Theta }}^{\rm{C}} \in {{\cal P}^{\rm{C}}}.
\end{align}
Similar to problem \eqref{max_power_gain}, the optimal solution of \eqref{max_power_gain_centralized} is derived as
\begin{align}\label{optimal_centralized}
{{{\bf{\tilde w}}}_k} = \frac{{{{\bf{a}}_M}\left( {\sin \theta _{{\rm{T}}}^{{\rm{AOD}}}} \right)}}{{\sqrt M }},{\left[ {{{\bf{\Theta }}^{\rm{C}}}} \right]_{n,n}} = \arg {\min _{\theta _n^{\rm{C}} \in {\cal F}}}\left| {{e^{j\theta _n^{\rm{C}}}} \!\!-\!\! {e^{j\theta _n^{{\rm{C}}**}}}} \right|,\forall n,
\end{align}
where $\theta _{n}^{{\rm{C}}**} = \arg \left( {{{\left[ {{\bf{\bar h}}_{r,k}^{\rm{C}}} \right]}_n}} \right) - \arg \left( {{{\left[ {{{\bf{a}}_{{\rm{S}}}}} \right]}_n}} \right)$. Accordingly, its optimal objective value is
\begin{align}\label{optimal_centralized}
{\gamma _k} = M{N^2}{\left( {\rho _g^{\rm{C}}} \right)^2}{\left( {\rho _{r,k}^{\rm{C}}} \right)^2}{\left( {\frac{{{2^b}}}{\pi }\sin \frac{\pi }{{{2^b}}}} \right)^2}.
\end{align}
Hence, the RHS of \eqref{rate_constrained_centralized} is upper-bounded by
\begin{align}\label{optimal_centralized}
{\log _2}\left( {1 + \frac{{\sum\nolimits_{k \in J} {{p_k}{{\left| {{{\left( {{\bf{h}}_k^{\rm{C}}\left( {{{\bf{\Theta }}^{\rm{C}}}} \right)} \right)}^H}{{{\bf{\tilde w}}}_k}} \right|}^2}} }}{{{\sigma ^2}}}} \right)\mathop  \le \limits^{\left( a \right)} {\log _2}\left( {1 + \frac{{\sum\nolimits_{k \in J} {{p_k}{\gamma _k}} }}{{{\sigma ^2}}}} \right)\mathop  \le \limits^{\left( b \right)} {\log _2}\left( {1 + \frac{{{P_{\max }}{\gamma _k}}}{{{\sigma ^2}}}} \right),
\end{align}
where (a) holds due to ${\left| {{{\left( {{\bf{h}}_k^{\rm{C}}\left( {{{\bf{\Theta }}^{\rm{C}}}} \right)} \right)}^H}{{{\bf{\tilde w}}}_k}} \right|^2} \le {\gamma _k}$ and (b) holds due to ${\gamma _1} =  \ldots  = {\gamma _K}$ based on \emph{Assumption 1} and $\sum\nolimits_{k \in {\cal J}} {{p_k}}  \le {P_{\max }}$.

Next, we show that ${\log _2}\left( {1 + \frac{{{P_{\max }}{\gamma _k}}}{{{\sigma ^2}}}} \right)$ is indeed achieved. By allocating weight of time ${\rho _k}$ for ${\Gamma _k}$ with $\sum\nolimits_{k \in {\cal J}} {{\rho _k} = 1}$, the sum-rate of user ${{\tilde u}_k}$'s, ${k \in {\cal J}}$, can be obtained as
\begin{align}\label{sum-rate_C}
\sum\nolimits_{k \in {\cal J}} {{\rho _k}} {\log _2}\left( {1 + \frac{{{P_{\max }}{\gamma _k}}}{{{\sigma ^2}}}} \right) = {\log _2}\left( {1 + \frac{{{P_{\max }}{\gamma _k}}}{{{\sigma ^2}}}} \right)
\end{align}
due to ${\gamma _1} =  \ldots  = {\gamma _K}$ and $\sum\nolimits_{k \in {\cal J}} {{\rho _k} = 1}$. Thus, we complete the proof.

\section*{Appendix C: \textsc{Proof of Theorem 1}}
Note that $R_s^{\rm{D}}$ is a function with respect to $K$ and thus we use $R_s^{\rm{D}}\left( K \right)$ to represent $R_s^{\rm{D}}$. By relaxing $K$ to a continuous variable, the first-order derivative of $R_s^{\rm{D}}\left( K \right)$ with respect to $K$ is given by
\begin{align}\label{first_order}
\frac{{\partial R_s^{\rm{D}}\left( K \right)}}{{\partial K}} = {\log _2}\left( {1 + {{{\gamma _0}} \mathord{\left/
 {\vphantom {{{\gamma _0}} {{K^3}}}} \right.
 \kern-\nulldelimiterspace} {{K^3}}}} \right) - 3 + \frac{3}{{1 + {{{\gamma _0}} \mathord{\left/
 {\vphantom {{{\gamma _0}} {{K^3}}}} \right.
 \kern-\nulldelimiterspace} {{K^3}}}}},
\end{align}
with

\begin{align}\label{gamma0}
{\gamma _0} = \frac{{{P_{\max }}M{N^2}{{\left( {\rho _{g,1}^{\rm{D}}} \right)}^2}{{\left( {\rho _{r,1}^{\rm{D}}} \right)}^2}}}{{{\sigma ^2}}}{\left( {\frac{{{2^b}}}{\pi }\sin \frac{\pi }{{{2^b}}}} \right)^2}.
\end{align}
Let ${{x = {\gamma _0}} \mathord{\left/
 {\vphantom {{x = {\gamma _0}} {{K^3}}}} \right.
 \kern-\nulldelimiterspace} {{K^3}}}$ and $g\left( x \right) \buildrel \Delta \over = \ln \left( {1 + x} \right) - 3 + 3/\left( {1 + x} \right)$. By further taking the first order derivative of $g\left( x \right)$ with respect to $x$, we obtain
\begin{align}\label{g_x_deriviation}
g'\left( x \right) \buildrel \Delta \over = \frac{{\partial g\left( x \right)}}{{\partial x}} = \frac{{x - 2}}{{{{\left( {1 + x} \right)}^2}}}.
\end{align}
From \eqref{g_x_deriviation}, we have $g'\left( x \right) < 0$ for $x \in \left( {0,2} \right)$ and $g'\left( x \right) \ge 0$ for $x \in \left[ {2,\infty } \right)$. Hence, $g\left( x \right)$ monotonously decreases with $x$ when $x \in \left( {0,2} \right)$ and $g\left( x \right)$ monotonously increases with $x$ when $x \in \left[ {2,\infty } \right)$. It can be readily verified that $g\left( x \right) < 0$ for $x \in \left( {0,2} \right]$ and $\mathop {\lim }\limits_{x \to \infty } g\left( x \right) =  + \infty  > 0$. Since $g\left( x \right)$ monotonously increases with $x$ for  $x \in \left[ {2,\infty } \right)$, equation $g\left( x \right) = 0$ has a single unique solution, denoted by ${{C_{{\rm{th}}}}}$, located in $\left( {0,\infty } \right)$. Hence, we have $g\left( x \right) \le 0$ for $x \in \left( {0,{C_{{\rm{th}}}}} \right]$ and $g\left( x \right) > 0$ for $x \in \left( {{C_{{\rm{th}}}},\infty } \right)$. Under the condition that ${\gamma _0} \le {C_{{\rm{th}}}}$, we have
\begin{align}\label{g_x_deriviation}
\left( {\partial R_s^{\rm{D}}\left( K \right)/\partial K} \right) = g\left( {{{{\gamma _0}} \mathord{\left/
 {\vphantom {{{\gamma _0}} {{K^3}}}} \right.
 \kern-\nulldelimiterspace} {{K^3}}}} \right)  \le  0, ~\forall 1 \le K \le M,
\end{align}
since ${{{\gamma _0}} \mathord{\left/
 {\vphantom {{{\gamma _0}} {{K^3}}}} \right.
 \kern-\nulldelimiterspace} {{K^3}}} \le {\gamma _0} \le {C_{{\rm{th}}}}$. In this case, ${R_s^{\rm{D}}\left( K \right)}$ monotonously decreases with $K$ and thus we have $R_s^{\rm{D}}\left( K \right)  \le  R_s^{\rm{D}}\left( 1 \right)$ for $\forall K \in \left\{ {1, \ldots ,M} \right\}$. Note that $R_s^{\rm{C}} = R_s^{\rm{D}}\left( 1 \right)$ and thus condition \eqref{sufficient_condition1} is obtained. By contrast, we have $g\left( {{{{\gamma _0}} \mathord{\left/
 {\vphantom {{{\gamma _0}} {{K^3}}}} \right.
 \kern-\nulldelimiterspace} {{K^3}}}} \right) \ge 0,\forall 1 \le K \le M$, if ${\gamma _0}/{M^3} \ge {C_{{\rm{th}}}}$ is satisfied. Hence, $R_s^{\rm{D}}\left( K \right)$ monotonously increases with $K$ for $1 \le K \le M$, which leads to $R_s^{\rm{D}}\left( K \right) \ge R_s^{\rm{D}}\left( 1 \right) = R_s^{\rm{C}}$,$\forall K \in \left\{ {1, \ldots ,M} \right\}$. The condition ${\gamma _0}/{M^3} \ge {C_{{\rm{th}}}}$ is equivalent to \eqref{sufficient_condition2} and thus we complete the proof.

\section*{Appendix D: \textsc{Proof of Proposition 3}}
We commence by expanding the equivalent channel as
\begin{align}\label{channel_reexpression}
{\left( {{\bf{h}}_{kl}^{\rm{D}}\left( {{\bf{\Theta }}_k^{\rm{D}}} \right)} \right)^H} \!\!=\!\! \rho _{r,kl}^{\rm{D}}\rho _{g,k}^{\rm{D}}\left(
\sqrt {\frac{{{\delta ^{\rm{D}}}}}{{{\delta ^{\rm{D}}} + 1}}} {\left( {{\bf{\bar h}}_{r,kl}^{\rm{D}}} \right)^H}{\bf{\Theta }}_k^{\rm{D}}{\bf{\bar G}}_k^{\rm{D}}
 + \sqrt {\frac{1}{{{\delta ^{\rm{D}}} + 1}}} {\left( {{\bf{\bar h}}_{r,kl}^{\rm{D}}} \right)^H}{\bf{\Theta }}_k^{\rm{D}}{\bf{\tilde G}}_k^{\rm{D}}
 \right).
\end{align}
Let ${\left( {{\bf{\bar h}}_{kl}^{\rm{D}}\left( {{\bf{\Theta }}_k^{\rm{D}}} \right)} \right)^H} = {\left( {{\bf{\bar h}}_{r,kl}^{\rm{D}}} \right)^H}{\bf{\Theta }}_k^{\rm{D}}{\bf{\bar G}}_k^{\rm{D}}$ and ${\left( {{\bf{\tilde h}}_{kl}^{\rm{D}}\left( {{\bf{\Theta }}_k^{\rm{D}}} \right)} \right)^H} = {\left( {{\bf{\bar h}}_{r,kl}^{\rm{D}}} \right)^H}{\bf{\Theta }}_k^{\rm{D}}{\bf{\tilde G}}_k^{\rm{D}}$. For the ${\rm{E}}\left\{ {\left\| {{\bf{h}}_{kl}^{\rm{D}}\left( {{\bf{\Theta }}_k^{\rm{D}}} \right)} \right\|_2^2} \right\}$, it can be derived as
\begin{align}\label{first_part}
\begin{array}{l}
{\rm{E}}\left\{ {\left\| {{\bf{h}}_{kl}^{\rm{D}}\left( {{\bf{\Theta }}_k^{\rm{D}}} \right)} \right\|_2^2} \right\}\\
 = {\left( {\rho _{r,kl}^{\rm{D}}\rho _{g,k}^{\rm{D}}} \right)^2}\left(
\frac{{{\delta ^{\rm{D}}}}}{{{\delta ^{\rm{D}}} + 1}}{\rm{E}}\left\{ {\left\| {{\bf{\bar h}}_{kl}^{\rm{D}}\left( {{\bf{\Theta }}_k^{\rm{D}}} \right)} \right\|_2^2} \right\}
 + \frac{1}{{{\delta ^{\rm{D}}} + 1}}{\rm{E}}\left\{ {\left\| {{\bf{\tilde h}}_{kl}^{\rm{D}}\left( {{\bf{\Theta }}_k^{\rm{D}}} \right)} \right\|_2^2} \right\}
 \right)\\
 + 4{\left( {\rho _{r,kl}^{\rm{D}}\rho _{g,k}^{\rm{D}}} \right)^2}\frac{{\sqrt {{\delta ^{\rm{D}}}} }}{{{\delta ^{\rm{D}}} + 1}}{\rm{E}}\left\{ {{\rm{Re}}\left( {{{\left( {{\bf{\bar h}}_{kl}^{\rm{D}}\left( {{\bf{\Theta }}_k^{\rm{D}}} \right)} \right)}^H}{\bf{\tilde h}}_{kl}^{\rm{D}}\left( {{\bf{\Theta }}_k^{\rm{D}}} \right)} \right)} \right\}\\
\mathop  = \limits^{\left( a \right)} {\left( {\rho _{r,kl}^{\rm{D}}\rho _{g,k}^{\rm{D}}} \right)^2}MN,
\end{array}
\end{align}
where (a) is obtained based on the results of ${\rm{E}}\left\{ {\left\| {{\bf{\bar h}}_{kl}^{\rm{D}}\left( {{\bf{\Theta }}_k^{\rm{D}}} \right)} \right\|_2^2} \right\} = {\rm{E}}\left\{ {\left\| {{\bf{\tilde h}}_{kl}^{\rm{D}}\left( {{\bf{\Theta }}_k^{\rm{D}}} \right)} \right\|_2^2} \right\} = MN$ and ${\rm{E}}\left\{ {{\mathop{\rm Re}\nolimits} \left( {{{\left( {{\bf{\bar h}}_{kl}^{\rm{D}}\left( {{\bf{\Theta }}_k^{\rm{D}}} \right)} \right)}^H}{\bf{\tilde h}}_{kl}^{\rm{D}}\left( {{\bf{\Theta }}_k^{\rm{D}}} \right)} \right)} \right\} = 0$. Then, ${{\rm{E}}\left\{ {{{\left| {{{\left( {{\bf{h}}_{kl}^{\rm{D}}\left( {{\bf{\Theta }}_k^{\rm{D}}} \right)} \right)}^H}\left( {{\bf{h}}_{ml'}^{\rm{D}}\left( {{\bf{\Theta }}_m^{\rm{D}}} \right)} \right)} \right|}^2}} \right\}}$ can be calculated as
\begin{align}\label{second_part}
\begin{array}{l}
{\rm{E}}\left\{ {{{\left| {{{\left( {{\bf{h}}_{kl}^{\rm{D}}\left( {{\bf{\Theta }}_k^{\rm{D}}} \right)} \right)}^H}\left( {{\bf{h}}_{ml'}^{\rm{D}}\left( {{\bf{\Theta }}_m^{\rm{D}}} \right)} \right)} \right|}^2}} \right\} = {\rm{E}}\left\{ {{{\left| {\sum\nolimits_{i = 1}^4 {{z_i}} } \right|}^2}} \right\} \\
= \sum\nolimits_{i = 1}^4 {{\mathop{\rm E}\nolimits} \left\{ {{{\left| {{z_i}} \right|}^2}} \right\}}  + \sum\nolimits_{i = 1}^4 {\sum\nolimits_{j = i+1}^4 {2{\mathop{\rm E}\nolimits} \left\{ {{\mathop{\rm Re}\nolimits} \left( {{z_i}z_j^*} \right)} \right\}} },
\end{array}
\end{align}
where
\begin{align}\label{second_part_appendix}
\begin{array}{l}
{z_1} = \frac{{{\delta ^{\rm{D}}}}}{{{\delta ^{\rm{D}}} + 1}}{\left( {{\bf{\bar h}}_{kl}^{\rm{D}}\left( {{\bf{\Theta }}_k^{\rm{D}}} \right)} \right)^H}{\bf{\bar h}}_{ml'}^{\rm{D}}\left( {{\bf{\Theta }}_m^{\rm{D}}} \right),{z_2} = \frac{{\sqrt {{\delta ^{\rm{D}}}} }}{{{\delta ^{\rm{D}}} + 1}}{\left( {{\bf{\bar h}}_{kl}^{\rm{D}}\left( {{\bf{\Theta }}_k^{\rm{D}}} \right)} \right)^H}{\bf{\tilde h}}_{ml'}^{\rm{D}}\left( {{\bf{\Theta }}_m^{\rm{D}}} \right), \\
{z_3} = \frac{{\sqrt {{\delta ^{\rm{D}}}} }}{{{\delta ^{\rm{D}}} + 1}}{\left( {{\bf{\tilde h}}_{kl}^{\rm{D}}\left( {{\bf{\Theta }}_k^{\rm{D}}} \right)} \right)^H}{\bf{\bar h}}_{ml'}^{\rm{D}}\left( {{\bf{\Theta }}_m^{\rm{D}}} \right),{z_4} = \frac{1}{{{\delta ^{\rm{D}}} + 1}}{\left( {{\bf{\tilde h}}_{kl}^{\rm{D}}\left( {{\bf{\Theta }}_k^{\rm{D}}} \right)} \right)^H}{\bf{\tilde h}}_{ml'}^{\rm{D}}\left( {{\bf{\Theta }}_m^{\rm{D}}} \right).
\end{array}
\end{align}
It can be readily shown that ${z_1} = 0$ according to the deployment condition unveiled in \eqref{deployment_condition}. For the remaining terms in \eqref{second_part}, we have
\begin{align}\label{second_part_appendix2}
\begin{array}{l}
{\rm{E}}\left\{ {{z_i}z_j^*} \right\} = 0,\forall i \ne j,\\
{\rm{E}}\left\{ {{{\left| {{z_2}} \right|}^2}} \right\} = {\rm{E}}\left\{ {{{\left| {{z_3}} \right|}^2}} \right\} = \frac{{{\delta ^{\rm{D}}}}}{{{{\left( {{\delta ^{\rm{D}}} + 1} \right)}^2}}}M{N^2},\\
{\rm{E}}\left\{ {{{\left| {{z_4}} \right|}^2}} \right\} = \frac{{M{N^2}}}{{{{\left( {{\delta ^{\rm{D}}} + 1} \right)}^2}}}.
\end{array}
\end{align}
Upon substituting \eqref{second_part_appendix2} into \eqref{second_part}, we arrive at
\begin{align}\label{second_part_result}
{\rm{E}}\left\{ {{{\left| {{{\left( {{\bf{h}}_{kl}^{\rm{D}}\left( {{\bf{\Theta }}_k^{\rm{D}}} \right)} \right)}^H}\left( {{\bf{h}}_{ml'}^{\rm{D}}\left( {{\bf{\Theta }}_m^{\rm{D}}} \right)} \right)} \right|}^2}} \right\}= {\left( {\rho _{r,kl}^{\rm{D}}\rho _{g,k}^{\rm{D}}} \right)^2}{\left( {\rho _{r,ml'}^{\rm{D}}\rho _{g,m}^{\rm{D}}} \right)^2}\frac{{2{\delta ^{\rm{D}}} + 1}}{{{{\left( {{\delta ^{\rm{D}}} + 1} \right)}^2}}}M{N^2}.
\end{align}
Based on \eqref{first_part} and \eqref{second_part_result}, \eqref{correlation1} can be obtained. Note that \eqref{correlation2} can be derived following similar steps, which are omitted for brevity.

\bibliographystyle{IEEEtran}

\vspace{-10pt}
\bibliography{IEEEabrv,myref}


\end{document}